\documentclass[aps,11pt,twoside, nofootinbib]{revtex4}

\usepackage[a4paper,top=3cm,bottom=2cm,left=2.5cm,right=2.5cm,marginparwidth=1.75cm]{geometry}
\usepackage{natbib}
\usepackage{amsmath}
\usepackage{graphicx}
\usepackage[colorinlistoftodos]{todonotes}
\usepackage[colorlinks=true, allcolors=blue]{hyperref}
\usepackage{bbm,microtype,mathrsfs,amssymb,color,amsthm,graphicx,cleveref,times,bm}
\usepackage{tikz-cd}
\usepackage{xcolor}
\usepackage{braket}
\usepackage{comment}
\usepackage{tikz}
\tikzstyle{tensor}=[rectangle,draw=blue!50,fill=blue!20,thick]
\usepackage{thmtools}
\usepackage{thm-restate}

\newtheorem{defn}{Definition}

\newtheorem{thm}{Theorem}[section] 
\newtheorem{cor}{Corollary}[section] 
\newtheorem{lem}{Lemma}[section] 
\newtheorem{prop}{Proposition}[section] 
\newtheorem{conj}{Conjecture}[section] 
\newtheorem{Question}{Question}[section]

\renewcommand{\>}{\rangle}
\newcommand{\ot}{\otimes}

\newcommand*{\cA}{\mathcal{A}}
\newcommand*{\cB}{\mathcal{B}}

\newcommand*{\cH}{\mathcal{H}}

\newcommand*{\cK}{\mathcal{K}}

\newcommand*{\cM}{\mathcal{M}}

\newcommand*{\cN}{\mathcal{N}}

\newcommand*{\cR}{\mathcal{R}}

\newcommand*{\cS}{\mathcal{S}}

\newcommand*{\tr}{\mathrm{Tr}}

\newcommand*{\id}{\rm id}

\def\bra #1{\langle #1\vert}
\def\ket #1{\vert #1\rangle}

\usepackage{color}
\definecolor{myred}{rgb}{1,0,0}

\definecolor{myblue}{rgb}{0,0,0.8}

\definecolor{myyellow}{rgb}{0.9,0.8,0}

\definecolor{mygreen}{rgb}{0,0.6,0}

\definecolor{myorange}{rgb}{0.6,0.6,0}

\definecolor{mycerul}{rgb}{0,0.6,1}

\newcommand{\CA}{\mathcal{A}}
\newcommand{\CB}{\mathcal{B}}

\newcommand{\BC}{\mathbb{C}}

\newcommand{\CE}{\mathcal{E}}

\newcommand{\CH}{\mathcal{H}}
\newcommand{\CI}{\mathcal{I}}

\newcommand{\CK}{\mathcal{K}}

\newcommand{\CM}{\mathcal{M}}
\newcommand{\CN}{\mathcal{N}}
\newcommand{\CO}{\mathcal{O}}
\newcommand{\CP}{\mathcal{P}}
\newcommand{\BP}{\mathbb{P}}

\newcommand{\CR}{\mathcal{R}}

\newcommand{\CS}{\mathcal{S}}

\newcommand{\CZ}{\mathcal{Z}}

\newcommand{\lV}{\lVert}
\newcommand{\rV}{\rVert}
\definecolor{KKgreen}{RGB}{0,100,70}
\definecolor{CFCblue}{RGB}{0,20,100}

\begin{document}

\title{Matrix Product Density Operators: when do they have a local parent Hamiltonian?}

\author{Chi-Fang (Anthony) Chen}

\affiliation{Institute for Quantum Information and Matter  \\ California Institute of Technology, Pasadena, CA 91125, USA}

\author{Kohtaro Kato}
\affiliation{Institute for Quantum Information and Matter  \\ California Institute of Technology, Pasadena, CA 91125, USA}
\affiliation{Center for Quantum Information and Quantum Biology,\\
Institute for Open and Transdisciplinary Research Initiatives,\\ Osaka University, Osaka 560-8531, Japan}

\author{Fernando G.S.L. Brand\~ao}
\affiliation{AWS Center for Quantum Computing, Pasadena, CA}
\affiliation{Institute for Quantum Information and Matter  \\ California Institute of Technology, Pasadena, CA 91125, USA}

\begin{abstract}
We study whether one can write a Matrix Product Density Operator (MPDO) as the Gibbs state of a quasi-local parent Hamiltonian. We conjecture this is the case for generic MPDO and give supporting evidences. To investigate the locality of the parent Hamiltonian, we take the approach of checking whether the quantum conditional mutual information decays exponentially. The MPDO we consider are constructed from a chain of 1-input/2-output (`Y-shaped') completely-positive maps, i.e. the MPDO have a local purification.  
We derive an upper bound on the conditional mutual information for bistochastic channels and strictly positive channels and show that it decays exponentially if the correctable algebra of the channel is trivial. 

We also introduce a conjecture on a quantum data processing inequality that implies the exponential decay of the conditional mutual information for every Y-shaped channel with trivial correctable algebra. We additionally investigate a close but nonequivalent cousin: MPDO measured in a local basis. We provide sufficient conditions for the exponential decay of the conditional mutual information of the measured states and numerically confirm they are generically true for certain random MPDO. 
 \end{abstract}
\maketitle
\tableofcontents
\section{Introduction}
Tensor networks provide useful ansatz for quantum many-body systems. 
In one-dimensional (1D) systems, the ground states of gapped local Hamiltonians can be efficiently approximated by Matrix Product States (MPS)~\cite{Affleck1987,Perez-Garcia2007,Hastings2007}. 
For the converse, generic MPS (which are technically called {\it injective} MPS) always have a gapped, local, and frustration-free {\it parent Hamiltonian} whose unique ground state is the MPS~\cite{Fannes1992,Perez-Garcia2007}. 
This correspondence between MPS and its parent Hamiltonian establishes a deep connection with gapped quantum systems, leading to the complete classification of 1D gapped quantum phases~\cite{Schuch2011}. For higher-dimensional systems, Projected Entangled Pair States (PEPS) are a natural generalization of MPS. PEPS have been used successfully to study gapped ground states~\cite{Affleck1987,verstraete2004renormalization}. Although the structural characterization of PEPS has not been completely established, local parent Hamiltonians also exist for injective and semi-injective PEPS~\cite{Molnar2018}. 

Matrix Product Density Operators (MPDO) are generalizations of MPS to describe 1D \textit{mixed} states. In Ref.~\cite{Hastings2006}, Hastings showed that any 1D Gibbs state of a local Hamiltonian (local Gibbs state in short) can be well-approximated by an MPDO with a polynomial bond dimension. This result justifies MPDO as a successful ansatz to study Gibbs states~\cite{Verstraete2004aa}.  

As a generic MPS is the ground state of the parent Hamiltonian, one may anticipate that a ``generic'' MPDO could be written as the Gibbs state of a local parent Hamiltonian. However, a fundamental obstacle to the latter question is the lack of a handy characterization of a generic MPDO. While Matrix Product Operators (MPO) are quite analogous to MPS, it is computationally hard to decide its positivity~\cite{PhysRevLett.113.160503}. Therefore, to make progress, one needs to hand-pick some parameterizable family of MPDO to begin with. 

Indeed, Cirac et al. analyzed MPDO at certain ``renormalization fixed-points", and showed that these fixed-point MPDO are Gibbs states of nearest-neighbor commuting Hamiltonians~\cite{Cirac2017}. 
Unfortunately, unlike MPS, a renormalization operation transforming a given MPDO to these fixed-point MPDO has not been well-defined yet. 

In this paper, we will focus on a physically motivated family: locally purifiable MPDO. These models arise from condensed matter physics, as they appear as the reduced state of MPS or a boundary state of PEPS. Technically, these models are chains of CP maps that connect to information theory. This structure is what enables much of our subsequent discussion.

Given the family of MPDO, our technical approach is to study whether the {\it conditional mutual information (CMI)} decays exponentially. The CMI $I(A:C|B)_\rho$ is a function defined for a tripartite state $\rho_{ABC}$ as
\begin{equation}
    I(A:C|B)_\rho:=S(AB)_\rho+S(BC)_\rho-S(B)_\rho-S(ABC)_\rho,
\end{equation}
where $S(A)_\rho=-\tr\rho_A\log_2\rho_A$ is the von Neumann entropy of the reduced state on $A$. 
Small values of CMI (with respect to certain tri-partitions of 1D system) turn out to be the necessary and sufficient condition  for being well-approximated by a local Gibbs state~\cite{KB16} (see Sec.~\ref{sec:prelimscmi} for more details). 

Now, we are in a position to present the guiding question of this work. Whenever the answer is affirmative, the parent Hamiltonians of the MPDO are (quasi-)local. 

\begin{Question} \label{question:main} For which MPDO does the CMI decay exponentially? Namely, for any tri-partition $ABC$ of the system where $B$ separates $A$ from $C$ with distance $\ell$ and constant $c>0$
\begin{align}
I(A:C|B) &= O(e^{-c\ell})\,?
\end{align}\label{main_decayCMI}
\end{Question}

In addition, we investigate the CMI of the MPDO after local measurements on the conditioning system $B$. As the CMI does not always decrease under measurement, it must be studied as an independent problem. The exponential decay of the CMI after measurement implies the outcome distribution is nearly a classical Markov distribution.

\begin{figure}[h]  
\includegraphics[width = 0.95\textwidth]{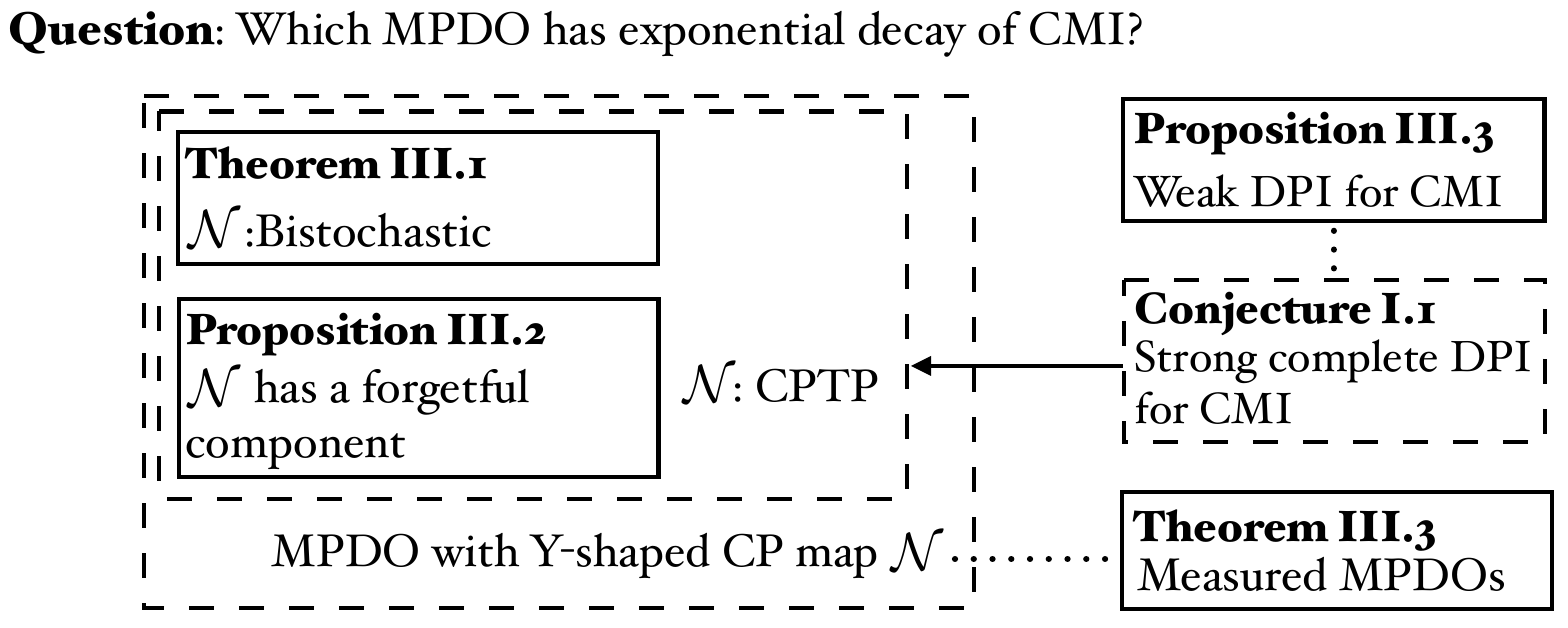}
\caption{The interplay between the main question, the main results, and the open conjectures. We ask which locally purifiable MPDO has exponential decay of CMI. When the MPDO is constructed by Y-shaped CPTP maps $\CN$, we obtain results for bistochastic channels and channels with a forgetful component. For general channels, we propose a strong data-processing inequality that implies decay of CMI but only prove a weaker version. When $\CN$ is only CP, we are only able to show the exponential decay of CMI for the measurement outcomes of MPDO. This figure omits the discussion on the trace norm CMI.}
\label{fig:MPDO_flow}
\end{figure}

\vspace{0.2in}
{\bf Main result. } In this work, we explore Question~\ref{question:main} in certain locally purifiable MPDO and point out missing technical conjectures that may lead to more general results.

We first study the quantum CMI of each tripartition of states generated by 1-input/2-output, ``Y-shaped" in short, channels. 
In particular, for bistochastic Y-shaped channels, we obtain an analytical bound for the CMI with a decay rate constant  (Theorem~\ref{thm: bistochastic}). The decay rate constant is strictly smaller than one if the channel has trivial correctable algebra, as defined in operator-algebra quantum error correction theory. We then generalize the argument for a slightly larger class of channels and derive a weaker bound on the CMI with another decay rate constant (Theorem~\ref{thm: partially invariant}). Therefore we show the exponential decay of the CMI (and a trace norm analog of CMI) for these MPDO.  We also prove the exponential decay of CMI for Y-shaped channels with a forgetful component which, in particular, include every strictly positive Y-shaped channels (Proposition~\ref{prop:forgetful}).

For general Y-shaped channels, we show that the CMI $I(A:C|B)$ must strictly decrease if one applies a channel with trivial correctable algebra on $C$ (Proposition~\ref{prop:strictDPI}). To bridge this result to Question~\ref{question:main}, we proposed a conjecture in the form of a data processing inequality for CMI with an explicit decay rate (Conjecture~\ref{ccDPI}), which would imply exponential decay of CMI (Proposition~\ref{prop:DPItoCMI}) for MPDO generated by such Y-shaped channels. The conjecture is known to be true for classical systems, however no result is known for quantum systems. 

We further study the CMI after local measurements on the conditional system $B$ in the computational basis for general locally purifiable MPDO (generated by completely positive (CP) Y-shaped maps)\footnote{ For technical reasons, the proof works for CP maps without the trace-preserving (TP) constraint, whereas in the unmeasured case it has to be a channel (CPTP-map). }. 
We prove two sufficient conditions, one is stronger than the other, that guarantee the exponential decay of the CMI of the measured MPDO (Theorem~\ref{mainthm}). 
We then provide a simple polynomial algorithm to check the stronger condition (Proposition~\ref{prop:algorithm}), and we numerically find the condition generically holds for MPDO generated by Y-shaped channels whose Stinespring unitaries are sampled from the Haar measure. 
 
\vspace{0.2 cm} 
 
{\bf Proof ideas.} The proof for the decay of CMI for bistochastic channel (Theorem~\ref{thm: bistochastic}) relies on decomposing the state into an uncorrelated state on $AB$ and $C$ plus a deviation which is traceless on $C$. We show the Hilbert-Schimdt norm of the deviation contracts under a bistochastic channel. In Theorem~\ref{thm: partially invariant}, we instead consider contraction of the deviation under the trace norm, using the associated tools for trace norm. For Y-shaped channels with a forgetful component (Proposition~\ref{prop:forgetful}), we use relative entropy convexity to show the CMI contracts at each step. The strict decay of the CMI for general channels (Proposition~\ref{prop:strictDPI})  follows from techniques in operator-algebra quantum error correction and properties of the Petz recovery map.

For the measured MPDO (Theorem~\ref{mainthm}) we are conditioning on a classical system $\bar{B}$, which implies the CMI $I(A:C|\bar{B})$ equals to the average of the mutual information for each outcome state $\rho_{AC,\bar{B}=b}$. The outcome states are constructed by sequential CP self-maps, and these are contractions in {\textit{Hilbert's projective metric}}. The CMI decays exponentially if the contraction ratio of every CP-map is strictly less than 1, which holds if the maps are all strictly positive (Condition 2). We also provide a stronger condition, Condition 1, which guarantees strict positivity after certain coarse-graining.

{\bf Structure of the paper.} 
We introduce basic concepts on MPDO and a few backgrounds in Section~\ref{preliminary}. In Section~\ref{main results}, we state our results and a conjecture on the data-processing inequality. The proofs are presented in Section~\ref{sec:proof} with several lemmas, whose detailed proof is left to Section~\ref{sec:proof of lemmas}.

\section{Preliminary} \label{preliminary}
In this section, we introduce basic notations and quantum information theoretical concepts that will be used in this paper. A quantum state (density operator) $\rho$ is a bounded operator on a finite-dimensional Hilbert space $\cH$ satisfying positivity $\rho\geq0$ and unit trace $\tr\rho=1$. We denote the set of quantum states on $\cH$ by $\cS(\cH)$.  Quantum systems are often denoted by capital letters $A,B,C,...$ and we abuse the same notation for the associated Hilbert spaces and the bounded operator spaces.  We denote the completely mixed state on a Hilbert space $B$ by $\tau_B$. The reduced state of $\rho$ associated to the system $A$ is denoted by $\rho_A$.  
$\|\cdot\|$ is the operator norm, $\|\cdot\|_p$ is the $p$-Schatten norm, $\|\cdot\|_{p-q}$ is the $p-q$ superoperator norm, and $\|\cdot\|_{cb}$ is the completely bounded 1-1 norm.

\subsection{Matrix Product Density Operators}
A general open boundary (uniform) MPDO $\rho\in\cS(\cH_n)$ is a quantum state $\rho$ written as 
\begin{align}
    \rho=\sum_{s_1,s'_1,...,s_n,s'_n}\<L|A^{s_n,s'_n}_n...A_1^{s_1,s'_1}|R\>\times |s_n...s_1\>\<s'_n...s'_1|\,,\label{eq:generalMPDO}
\end{align}
where each $\{A_k^{s_is'_i}\}_{s_i,s'_i}$ is a set of  $D$-dimensional matrices and $|L\>,|R\>$ are $D$-dimensional vectors with a constant $D$. Here, $A_k^{s_1s'_1}, |L\>$ and $|R\>$ are chosen so that positivity $\rho\geq0$ is satisfied for arbitrary $n$. This is a non-trivial condition and in general it is computationally hard to determine whether a given MPO is positive or not~\cite{PhysRevLett.113.160503}. In this paper, we will specialize to specific sub-classes of MPDO to guarantee the positivity of the state.

\begin{figure}[htbp]  
\begin{tikzpicture}[inner sep=1mm]
    \foreach \i in {1,...,5} {
        \node[tensor] (\i+0) at (-\i, 1) {$A_\i$};
        \node (\i spin) at (-\i, 0.3) {$\bra{s'_\i}$};
        \node (\i spinu) at (-\i, 1.7) {$\ket{s_\i}$};
        \node (\i tensor) at (-\i, 0.82) {};
        \draw[-] (\i+0) -- (\i spin);
	\draw[-] (\i+0) -- (\i spinu);
    };
    \foreach \i in {1,...,4} {
        \pgfmathtruncatemacro{\iplusone}{\i + 1};
       \draw[-] (\i+0) -- (\iplusone +0);
    };
      \node (spinl1) at (-6, 1) {$\bra{L}$};
      \draw[-] (5+0) -- (spinl1);
      \node (spinr0) at (0, 1) {$\ket{R}$};
      \draw[-] (1+0) -- (spinr0); 
\end{tikzpicture}

\caption{The tensor network representation for general MPDO~\eqref{eq:generalMPDO}. It differs from MPS that both bra and ket indices are distinguished, describing an operator instead of a vector. 
}
\label{MPDO_graphic}
\end{figure}
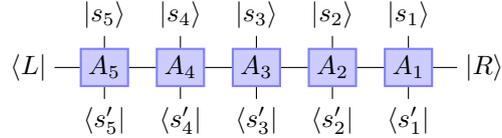

Throughout this paper, we consider MPDO constructed by concatenating Completely-Positive (CP) maps, i.e., locally purifiable MPDO. Let $\cH, \cK$ be finite-dimensional Hilbert spaces and $\cN:\cB(\cH)\to\cB(\cK\ot\cH)$ be a 1-input/2-output CP-map where one of the two output systems has the same dimension as the input. We interchangeably refer such a map as a {\it Y-shaped} map. For some initial state $\sigma$ and $\ell\in\mathbb{N}$, we obtain a (possibly unnormalized) state $\rho^\sigma_{\ell}(\CN)\in\cS(\cK^{\ot \ell}\ot\cH)$ defined by
\begin{equation}\label{def:family1}
    \rho^\sigma_\ell(\CN):=\left(\id_{\cK^{\ot (\ell-1)}}\ot\cN\right)\circ...\circ\left(\id_\cK\ot\cN\right)\circ\cN[\sigma],
\end{equation}
where $\id_\cK$ is the identity map on $\CK$. We mostly abbreviate $\rho_\ell^\sigma(\cN)$ as $\rho_\ell^\sigma$. This kind of states are classified as a finite-dimensional $C^*$-finitely correlated states~\cite{Fannes1992}~\footnote{Eq.~\eqref{def:family1} contains states without the consistency constraint imposed on the input state  $\tr_B\rho^\sigma_{BC}=\sigma$, which is required for $C^*$-infinite finitely correlated states.}.
By using the Kraus representation $\cN(X)=\sum_iT_iXT_i^\dagger$, one can easily verify $\rho_\ell^\sigma$ is a MPDO. 

The normalization $\tr\rho^\sigma_\ell=1$ is guaranteed when $\sigma$ is a quantum state and $\CN$ is a quantum channel, i.e., CP and Trace-Preserving (CPTP) map. 
Since 1D local Gibbs states always have exponentially decaying two-point correlation~\cite{Araki1969}, we are interested in  Eq.~\eqref{def:family1} with finite correlation length. The MPDO has a finite correlation length if $\tr_\cK\circ\cN$ has unique maximum eigenvalue 1, and especially the correlation length is exactly zero when $\tr_\cK\circ\cN$ is a constant channel (Fig.~\ref{fig:chain of Y-shaped}). The choice of $\sigma$ is rather arbitrary in finite systems. For convenience we mainly choose $\sigma$ as one side of the maximally entangled state, and denote the corresponding state on $\cH\ot\cK^{\ot \ell}\ot \cH$ by $\rho_\ell$. We often denote $A$ and $C$ to be two systems $\cH$ at the end, and $B_i\, (i=1,...,\ell)$ to be the remaining systems $\cK$.  Note that we can recover arbitrary $\rho^\sigma_\ell$ from $\rho_\ell$ by applying a suitable positive operator on $A$, which is the other side of the maximally entangled state $\sigma_{A\bar{A}}$, and then trace out $A$. 
\begin{figure}[htbp]  
\includegraphics[width = 0.7\textwidth]{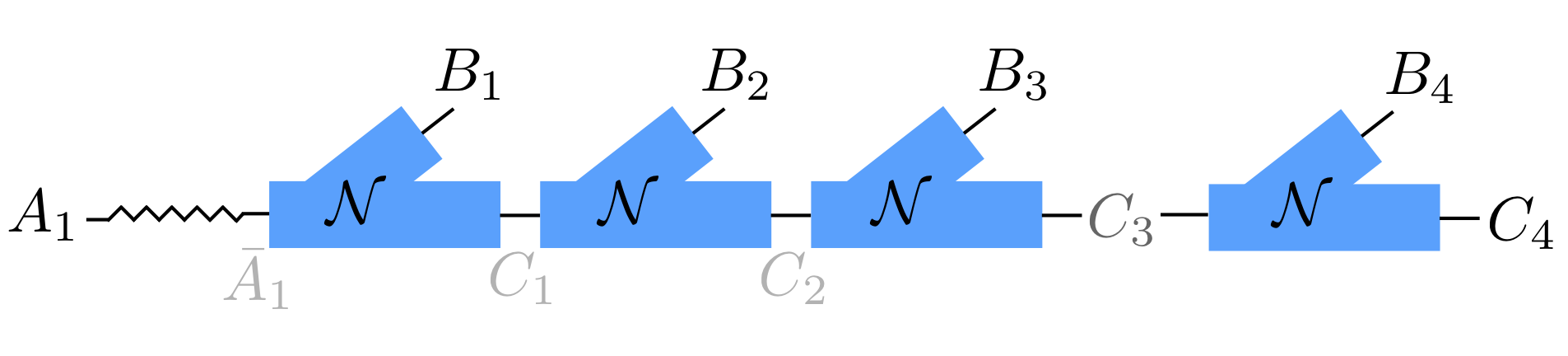}
\includegraphics[width = 0.7\textwidth]{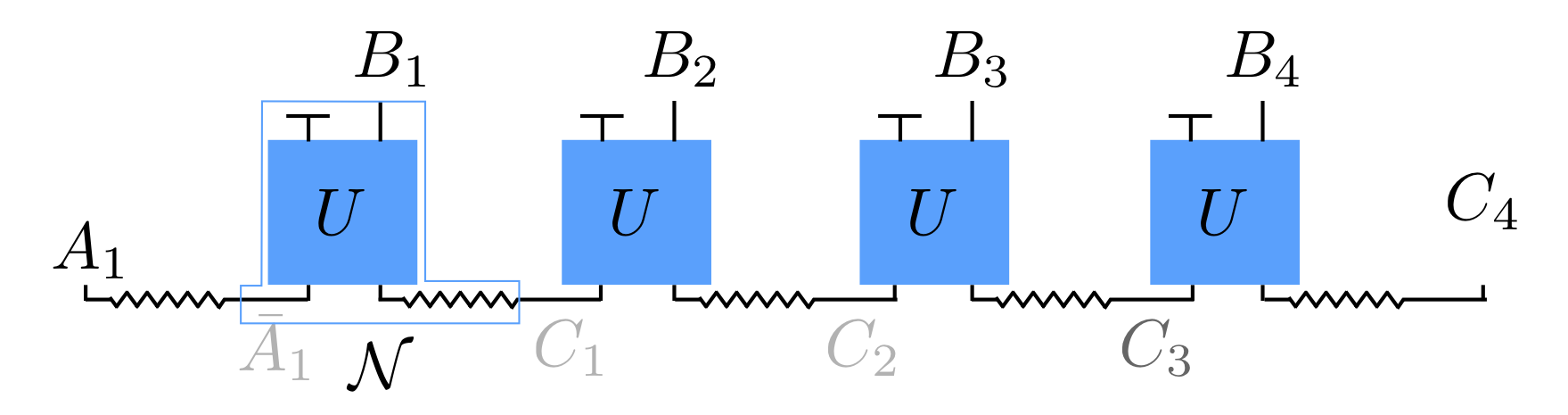}

\caption{(Top) The chain of Y-shaped channel $\CN$. The zigzag line represents a maximally  entangled pair. When we concatenate another channel to system $AB_1B_2B_3C_3$, subsystem $C_3$ is mapped to system $B_{4}C_{4}$.  (Bottom) Y-shaped channel with zero correlation length can be generated by unitaries or isometries and maximally entangled pairs, where tracing out any system $B_k$ makes the left and right segment uncorrelated.}
\label{fig:chain of Y-shaped}
\end{figure}

When we further perform a projective measurement on $B$ to Eq.~\eqref{def:family1} in the computational basis $\{|s_1s_2...s_l\>_B\}$, the output subsystem ${\bar B}$ becomes entirely classical (Fig.~\ref{fig:measured MPDO}). The corresponding Y-shaped map $\CN$ can then be decomposed into 
$\CN[\cdot] = \sum_{s_k} \ket{s_k}\bra{s_k}\otimes  \CM_{s_k}[\cdot]$, where $\CM_s[\cdot]$ is a CP self-map defined as
 \begin{align}\label{eq:CP self-maps}
\CM_{s_k}[\rho]&:=\<s_k|_{B_k}\left(\cN[\rho]\right)|s_k\>_{B_k}\,.
 \end{align}
Note that if $\CN$ is TP, then $\sum_s\CM_s$ is a CPTP-map and thus $\{\CM_s\}_s$ forms a quantum instrument.

 We can rewrite the measured MPDO by 
 \begin{align}\label{eq:measured MPDO}
 \rho_{A\bar{B}C}  = \sum_b p(b)\ket{b}\bra{b}_{\bar B} \otimes \rho_{AC, b}, 
 \end{align}
 where each $\rho_{AC, b}$ is an output state with a particular measurement outcome $b=\{s_1,...,s_l\}$
\begin{equation}
 \rho_{AC, b} =\frac{M_b [\sigma_{A\bar{A}}]}{\tr(M_b[\sigma_{A\bar{A}}])}:=\frac{\CM_{s_{\ell}}\cdots \CM_{s_1} [\sigma_{A\bar{A}}]}{\tr(\CM_{s_{\ell}}\cdots \CM_{s_1} [\sigma_{A\bar{A}}])},
 \end{equation}
 with probability 
 \begin{equation}
 p(b) =  \tr(\CM_{s_{\ell}}\cdots \CM_{s_1} [\sigma_{A\bar{A}}]).
 \end{equation}
\begin{figure}[ht]  
\includegraphics[width = 0.7\textwidth]{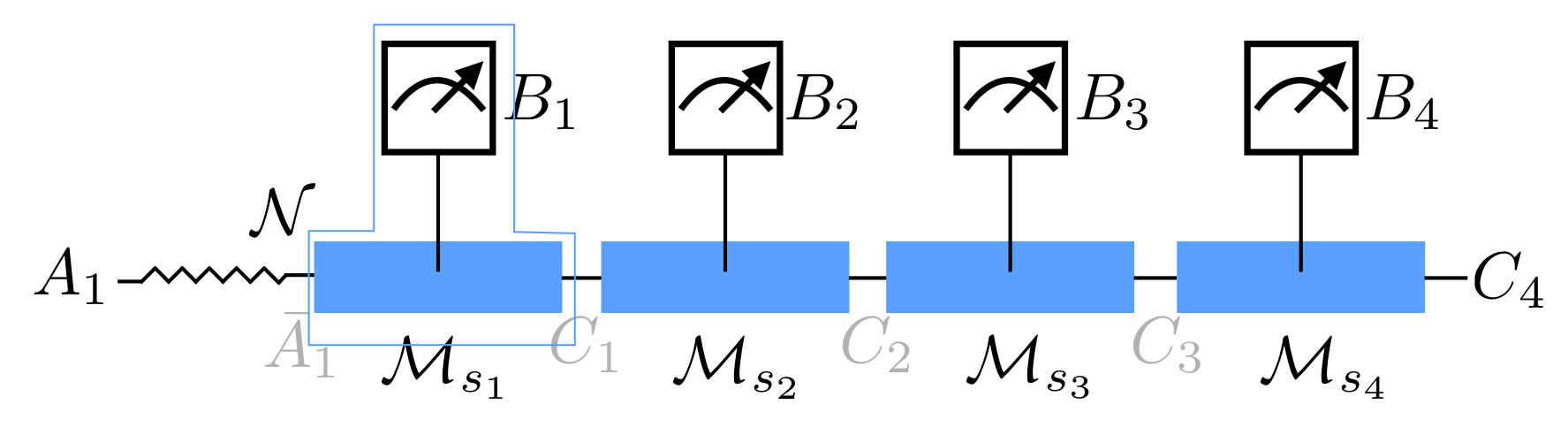}
\caption{The chain of Y-shaped CP maps reduced to sequences of CP self-maps upon measurement outcomes on $B_1, \ldots, B_4$. It mimics a classical hidden Markov chain, with the extra bit of the quantumness confined in system $A$ and $C$.}
\label{fig:measured MPDO}
\end{figure}

\subsection{Strong data-processing inequality constants}\label{sec:prelimsdpi}
The quantum relative entropy $D(\rho\|\sigma):=\tr\rho\log_2(\rho-\sigma)$ 
is a distance-like measure between two quantum states $\rho,\sigma$. 
One crucial feature of the relative entropy is that it obeys the {\it data-processing inequality (DPI)}: for any states $\rho, \sigma$ and any CPTP-map $\CE$, it follows that~\cite{Linbald1975,Uhlmann:1977aa}
\begin{align}
    D\left(\CE[\rho]\|\CE[\sigma]\right)\leq D(\rho\|\sigma). 
\end{align} 
The above DPI implies the monotonicity of the mutual information $I(A:C)_\rho :=D(\rho_{AC}\|\rho_A\ot\rho_C)$ and the conditional mutual information $I(A:C|B)_\rho$ under a local CPTP-map $\CE_C:C\to C'$, that is,
\begin{align}
    I(A:C')_{\CE_C(\rho_{AC})}&\leq I(A:C)_{\rho_{AC}}\,,\\
    I(A:C'|B)_{\CE_C(\rho_{ABC})}&\leq I(A:C|B)_{\rho_{ABC}}\,.
\end{align}
The equality holds if and only if there is a CPTP-map $\cR_{C'}:{C'\to C}$ such that $\CR\circ\CE[\rho]=\rho$~\cite{Petz1986,Petz1988}. 

{\it Strong DPI constants}, or the {\it contraction coefficients} are multiplicative factors in the above monotonicity inequalities. For the mutual information and CMI, they are defined as
\begin{align}
    \eta_{AC}(\CE_C)&:=\sup_{\rho_{AC}}\frac{I(A:C')_{\CE_C[\rho_{AC}]}}{I(A:C)_{\rho_{AC}}}\,,\\
    \eta_{ABC}(\CE_C)&:=\sup_{\rho_{ABC}}\frac{I(A:C'|B)_{\CE_C[\rho_{ABC}]}}{I(A:C|B)_{\rho_{ABC}}}\,,
\end{align}
where the values are bounded as $0\leq \eta_{AC}(\CE_C),\eta_{ABC}(\CE_C)\leq1$ by monotonicity. Note that $\eta_{ABC}(\CE_C)$ reduces to $\eta_{AC}(\CE_C)$ by taking $B$ to be trivial. The strong DPI constants bound how much correlations are preserved after applying the channel on subsysem $C$.

In the classical case, the strong DPI constants for different quantities are equivalent.
\begin{thm}[{Tight contractive DPI for classical mutual information~\cite[Theorem~4]{anantharam2013maximal} }]\label{thm:classicalSDPIc}
For any probability distribution $p_{AC}$ and any classical channel $\CE:C\to C'$, 
\begin{equation}
    \eta_{AC}(\CE)=\eta_{ABC}(\CE)=s^*(\CE)\,,
\end{equation}
where $s^*(\CE)$ is the strong DPI constant for the relative entropy on $C$
\begin{equation}
    s^*(\CE):= \sup_{\substack{p,q\\ p\neq q}} \frac{D\left(\CE\left[q_C\right]\right.\left\|\CE\left[p_C\right]\right)}{D\left(p_C\right.\left\|q_C\right)}\,.
\end{equation}
\end{thm}

Remarkably, $\eta_{AC}(\CE)$ and $\eta_{ABC}(\CE)$ are independent of $A, B$ in the classical case. 
The crucial point is that the classical mutual information and the CMI are functions of conditional probability distribution on $C$. 
These functions are then written by convex combinations in classical systems, and regardless of how large auxilary classical system is involved, the Carath\'eodory theorem implies the auxilary system can always be reduced to dimension $|C|+1$~\cite{Danninger-Uchida2001}.
 
Unfortunately, no analog of conditional distribution nor the Carath\'eodory-like cardinality bound has been found for quantum systems (see e.g., Ref.~\cite{6555753} for a discussion on this problem). Therefore the classical approach failed to work in quantum regime. 
For this reason, for quantum systems it remains open whether Theorem~\ref{thm:classicalSDPIc} also holds, or whether the strong DPI constants are independent of the size of the auxiliary systems or not.

\subsection{The conditional mutual information and Gibbs states}\label{sec:prelimscmi}
The CMI and Gibbs states are intimately connected. In classical systems, the Hammersley-Clifford theorem~\cite{HCthm} states that a (positive) probability distribution $p_{XYZ}$ is a Gibbs distribution
\begin{equation}\label{eq:classicalGibbs}
    p_{XYZ}(x,y,z)=\frac{e^{-h_{XY}(x,y)-h_{YZ}(y,z)}}{Z}\,,
\end{equation}
where $Z$ is the normalization constant, if and only if $I(X:Z|Y)_p=0$. Moreover, CMI can be written as
\begin{align}\label{eq:classicalCMI}
    I(X:Z|Y)_p=\min_{q:I(X:Z|Y)_q=0}D(p_{XYZ}\|q_{XYZ})\,,
\end{align}
and thus the small value of the CMI implies the state is close to a Gibbs state~\eqref{eq:classicalGibbs}. These results are naturally extended to 1D classical spin chains.

Although Eq.~\eqref{eq:classicalCMI} does not hold in general quantum systems, Ref.~\cite{KB16} shows the following bound:
\begin{thm}[Theorem 1,\cite{KB16}]\label{thm:KB16} Let $\rho_{A_1...A_n}$ be a state satisfying $I(A_1...A_{k-1}:A_k...A_n|A_k)\leq\varepsilon$. Then
there exists a local Hamiltonian $H=\sum_{i=1} h_{A_i,A_{i+1}}$ with $h_{A_i,A_{i+1}}$ only acts on $A_iA_{i+1}$, such that the relative entropy is controlled by 
\begin{equation}
     D\left(\rho\left\|\frac{e^{-H}}{\tr e^{-H}}\right.\right)\leq n\varepsilon\,.
\end{equation}
\end{thm}
If $\varepsilon=0$, we recover the quantum Hammersley-Clifford theorem~\cite{Hayden2004,brown2012quantum} and the Hamiltonian $H$ is the sum of terms like $-\log\rho_{A_i,A_{i+1}}$.  

When MPDO has exponentially decaying CMI for appropriate tripartitions, then we first coarse-grain sufficiently many (possibly logarithmic w.r.t. $1/\varepsilon$) neighbouring sites as one site and then apply Theorem~\ref{thm:KB16} to show that the MPDO is well-approximated by a local Gibbs state. 
\begin{cor}
Suppose $\rho_{B_1...B_m}$ satisfies exponential decay of CMI for every tripartition 
\begin{align}
I\left(B_1\cdots B_n: B_{n+\ell+1}\cdots B_{m}|B_{n+1}\cdots B_{n+\ell} \right) \le \CO( e^{- c \ell } )   .
\end{align}
Then coarse graining $\ell = \Omega\left( \frac{1}{c} \log(m/\epsilon) \right)$ sites ensures (e.g., $B_1\cdots B_\ell := A_1$ ) there exist a (quasi-)local Hamiltonian $H=\sum_{i=1} h_{A_i,A_{i+1}}$ such that
\begin{align}
    D\left( \rho_{A_1...A_m}|| \frac{ e^{-H}}{\tr[e^{-H}]}\right) \le \epsilon.
\end{align}
\end{cor} 
Note that relative entropy estimate converts to fidelity by
\begin{align}
    D(\rho||\sigma) \ge - 2\log(F(\rho,\sigma)).
\end{align}
\subsection{The correctable algebra and recovery channels}
To characterize what information is left undisturbed by a channel, we use the theory of operator-algebra quantum-error correction~\cite{Beny2007,Beny2007a}. For a given channel $\CE:A\to A'$, we define the {\it correctable algebra} $\CA(\CE)\subset \CB(\CH_A)$ as 
\begin{align}\label{def:correctable algebra}
    \CA(\CE):={\rm Alg}\left\{O_A \left| \left[O_A, E_a^\dagger E_b\right]=0\,,\;\;\forall a,b\right. \right\}.
\end{align}
$\CA(\CE)$ is a finite-dimensional $C^*$-algebra containing all observables whose information is perfectly preserved under $\CE$. One can always perfectly recover the information associated to an observable in the correctable algebra: there exists a CPTP-map $\CR:A'\to A$ called {\it recovery map} such that 
\begin{align}
    \CE^\dagger\circ\CR^\dagger(O_A)=O_A\,,\quad \forall O_A\in \CA(\CE)\,.
\end{align}
This implies that for any $O_A\in\cA(\CE)$, there always is a corresponding observable $\CR^\dagger(O_A)$ on $A'$ such that $\tr(O_A\rho_A)=\tr\left(\CR^\dagger(O_A)\CE(\rho_A)\right)$ for all $\rho_A$. 
The condition in Eq.~\eqref{def:correctable algebra} is a generalization of the well-known Knill-Laflamme condition~\cite{PhysRevLett.84.2525} in the standard quantum error correction theory, in which $\CA(\CE)=\CB(\CH_{code})$ for a protected code subspace $\CH_{code}$.

A common candidate for the recovery channel is the Petz map:
\begin{defn}The Petz recovery map $\CP_{\sigma,\CE}$ with the reference state $\sigma$ and the target channel $\CE$ is defined as
\begin{equation}\label{eq:Petz}
\CP_{\sigma,\CE}[\cdot] :=\sigma^{1/2}\CE^\dagger\left[\CE(\sigma)^{-1/2}(\cdot)\CE(\sigma)^{-1/2}\right]\sigma^{1/2}.
\end{equation}
\end{defn}
As it only features moderately in this paper, we save the vast literature on universal recovery channel elsewhere (see e.g.,~\cite{Junge_2018}), and include the relevant ideas  on invariant subspaces in Appendix~\ref{sec:Petz and correctable}. This yields another characterization of the correctable algebra handy for our purposes.
\begin{prop}\label{invar_cor}
The correctable algebra of a channel $\CE$ equals to the fixed-point algebra of the Petz-recovery channel composed with channel $\CP_{\tau, \CE}\circ\CE$ (and the dual $\CE^\dagger \circ \CP_{\tau, \CE}^\dagger$, due to being self-adjoint).
\begin{equation}
S_{\CE^\dagger \circ \CP_{\tau, \CE}^\dagger} = S_{\CP_{\tau, \CE}\circ\CE} = \CA(\CE)    
\end{equation}
\end{prop}
\section{Main results}\label{main results}
Here we state our main results in a concrete way. 

\subsection{Decay of CMI for bistochastic Y-shaped channels}\label{sec:bistochastic}
Our first main result is showing that any state constructed from a bistochastic~\footnote{Usually, bistochastic refers to a self-map channel in the literature. 
Here we have higher output dimension than input, but we abuse the terminology in this paper.}  Y-shaped channel (CPTP-map) is exponentially close to a product state on $AB$ and $C$. 
This provides an upper bound of the CMI (and also the trace norm CMI in Sec.~\ref{sec:trace norm CMI}). See Sec.~\ref{sec: sketch bistochastic} for the proof. 
\begin{thm} \label{thm: bistochastic}
For a bistochastic (unital) 
Y-shaped channel $\CN[\tau_{C_1}] = \tau_{B_1}\otimes \tau_{C_2}$, consider a family of states defined by
\begin{equation}
    \rho_{AB_1\cdots B_{\ell} C_{\ell}}:=\left(\id_\cH\ot\id_\cK^{\ot (\ell-1)}\ot\cN\right)\circ...\circ\left(\id_\cH\ot\id_{\cK}\ot\cN\right)\circ\cN[\sigma_{A\bar{A}}],
\end{equation}
where $\sigma_{A\bar{A}}$ is the maximally entangled state. 
Then, it holds that 
\begin{equation}
\lV \rho_{AB_1\cdots B_{\ell} C_{\ell}} - \rho_{AB_1\cdots B_{\ell}} \otimes \tau_C \rV_1 = O(\eta^\ell).\label{eq:1 close to pd}
\end{equation}
This implies the CMI is bounded as
\begin{align}
I(A: C|B_{1}\cdots B_{\ell})_{\rho_l} &=  O( \ell \eta^{\ell} ).
\end{align}
The contraction ratio $\eta$ is given as
\begin{equation}
\eta: = \limsup_{\rho_C \rightarrow \tau_C } \frac{D(\CN[\rho_C]\|\tau_{B_1C})}{D(\rho_C\|\tau_C)},
\end{equation}
where $\eta<1$ if and only if $\CN$ has trivial correctable algebra. Furthermore, $\eta$ coincides with the second eigenvalue of channel $\CZ:= \CP_{\tau,\,\CN}\circ\CN$, where $\CP_{\tau,\,\CN}$ is the Petz recovery map~\cite{10.1093/qmath/39.1.97} as in Eq.(\ref{eq:Petz}) with the completely mixed state $\tau$ as reference.
\end{thm}
\noindent{\bf Remark.} Here we only consider a particular tripartition $ABC$ to show the decay of the CMI. However, under additional moderate assumptions, the exponential decay of CMI of this tripartition implies the CMI is also small for other tripartitions like $A':=AB_1...B_n, B':=B_{n+1}...B_{n+m}, C':=B_{n+m+1}...B_\ell C_\ell$. To see this, we use the chain rule of CMI: $I(A':C'|B')_\rho=I(A:C_\ell|B_1...B_\ell)_\rho+I(A:B_{n+m+1}...B_\ell|B_1...B_{n+m})_\rho+I(B_1...B_n:B_{n+m+1}...B_\ell C_\ell|B_{n+1}...B_{n+m})_\rho$. The second term is upper bounded by $I(A:C_{n+m}|B_1...B_{n+m})_\rho$ by the monotonicity of the CMI, and thus the first and the second terms obey the exponential decay as $m$ grows. Bounding the third term may require an additional assumption. One way to bound the last term is by explicit calculation. For instance, it is 0 for the bistochastic case since $\rho_{BC}=\tau_{BC}$. Another way to bound it is assuming the zero correlation length condition (bottom, Fig.~\ref{fig:chain of Y-shaped}), under which the state can be generated by the inverse channel ${\tilde \CN}:C\to BA$ as well as $\CN$. This guarantees $I(B_1...B_n:B_{n+m+1}...B_\ell C_\ell|B_{n+1}...B_{n+m})_\rho$ is upper bounded by $I(A_n:C_{n+m}|B_{n+1}...B_{n+m})_\rho$, which is the CMI for $\rho_m$ up to the translation shift.

In this paper, we pay less attention to the case of periodic MPDO where output $C$ of channel got fed into the input $A$, as this gluing step might be tricky. In the bistochastic case this is tractable (Appendix~\ref{sec:periodic bistochastic MPDO}) and the state becomes even simpler as the maximally mixed state plus an exponentially decaying global operator. There, the parent Hamiltonian would be trivial in the thermodynamic limit.

A crucial property of bistochastic Y-shaped channels is they act as a contraction onto certain operator subspace. This implies
Eq.~\eqref{eq:1 close to pd} at large $\ell$, and then the CMI bound follows from the continuity of the entropy. An example of bistochastic Y-shaped channel is given as follows. 

\noindent{\bf Example:} A bistochastic $Y-$shaped channel $\CE_p:\cB({\mathbb C}^2)\to \cB({\mathbb C}^2\ot{\mathbb C}^2)$ defined by Kraus operators
\begin{align}
    K_1=\left[
\begin{array}{cc}
\frac{\sqrt{1-p}}{2} &\frac{\sqrt{1-p}}{2} \\
 0&0 \\
 \frac{\sqrt{1-p}}{2}&-\frac{\sqrt{1-p}}{2} \\
 0&0\\
\end{array}
\right],
K_2=\left[
\begin{array}{cc}
\frac{\sqrt{p}}{2} &\frac{\sqrt{p}}{2} \\
 0&0 \\
 -\frac{\sqrt{p}}{2}&\frac{\sqrt{p}}{2} \\
 0&0\\
\end{array}
\right],
K_3=\left[
\begin{array}{cc}
0&0\\\frac{\sqrt{1-p}}{2} &\frac{\sqrt{1-p}}{2} \\
 0&0 \\
- \frac{\sqrt{1-p}}{2}&\frac{\sqrt{1-p}}{2} \\
\end{array}
\right],
K_4=\left[
\begin{array}{cc}
0&0\\\frac{\sqrt{p}}{2} &\frac{\sqrt{p}}{2} \\
 0&0 \\
\frac{\sqrt{p}}{2}&-\frac{\sqrt{p}}{2} \\
\end{array}
\right],
\end{align}
has exponentially decaying CMI for $p\in(0,1/2)\cup(1/2,1)$. Otherwise,  $I(A:C|B)_{\rho_\ell}=0$  for $p=1/2$ and $I(A:C|B)_{\rho_\ell}=1$ for $p=0,1$. 

While being bistochastic is crucial to provide a decay rate defined by relative entropy in Theorem~\ref{thm: bistochastic}, we can work with slightly general channels with a likely suboptimal decay rate. See Sec.~\ref{sec: sketch partially invariant} for the proof. 
\begin{thm} \label{thm: partially invariant}
For a channel such that $\CN[\nu_{C_1}] = \sigma_{B_1}\otimes \nu_{C_2}$ for some $\sigma$ and $\nu$, and the maximally entangled input state $\sigma_{A\bar{A}}$, the resulting chain
\begin{equation}
    \rho_{AB_1\cdots B_{\ell} C_{\ell}}:=\left(\id_\cH\ot\id_\cK^{\ot (\ell-1)}\ot\cN\right)\circ...\circ\left(\id_\cH\ot\id_{\cK}\ot\cN\right)\circ\cN[\sigma_{A\bar{A}}]
\end{equation}
satisfies
\begin{equation}
\lV \rho_{AB_1\cdots B_{\ell} C_{\ell}} - \rho_{AB_1\cdots B_{\ell}} \otimes \nu_C \rV_1 = O(\eta^\ell).
\end{equation}
This implies the CMI is bounded as
\begin{align}
I(A: C|B_{1}\cdots B_{\ell})_{\rho_\ell} &=  O( \ell \eta^{\ell} ).
\end{align}
The contraction ratio $\eta$ is given as 
\begin{equation}
    \eta := \min (1, 16 d_C \eta_{1,C}),
\end{equation}
where $\eta_{1,C}$ is the trace norm contraction ratio defined on system C only
\begin{equation}\label{contract_ratio}
\eta_{1,C}:=\sup_{\rho_C,\rho'_C}\frac{\lVert\CE_C[\rho_{C}]-\CE_C[\rho'_{C}]\lVert_1 }{\lVert\rho_{C}-\rho'_{C}\lVert_1}.
\end{equation}
\end{thm}
The decay of CMI for more general tripartition follows from the same argument as the bistochastic case (Theorem~\ref{thm: bistochastic} ). The distinction, though, is that our bound for the contraction ratio suffers from extra factor of dimension $d_C$ due to conversion to the completely bounded norm (Lemma~\ref{lem:cb to 1-1}), so the bound is meaningful only for the noisy regime $\eta_{1,C}<(16d_c)^{-1}$.

\subsubsection{The trace norm CMI}\label{sec:trace norm CMI}
In some cases, the trace norm (or the trace distance) has more applicable tools than the relative entropy. We now introduce a trace-norm variant of CMI.
 We expect it to help understanding the qualitative behavior of the CMI.
\begin{defn}[The Trace norm CMI] \label{def:trace norm CMI}
\begin{equation}
I_1(A:C|B) := \lV \rho_{ABC} -\rho_{A}\ot\rho_{BC} \rV_1 -\lV \rho_{AB} -\rho_{A}\ot\rho_{B} \rV_1
\end{equation}
\end{defn}
The definition is motivated by the form of CMI as $I(A:C|B)=I(A:BC)-I(A:B)$ and replace each mutual information by the trace distance (without the normalization for the simplicity). The mutual information and the trace distance bound each other (in a non-linear way) by the quantum Pinsker inequality~\cite[Theorem ~11.9.1]{wilde2011classical} and the continuity of entropy, and thus we expect this trace-norm variant of CMI partially captures common characteristics of the states as the CMI. Indeed, in the special cases analyzed in the following two theorems (Theorem~\ref{thm: bistochastic}, Theorem~\ref{thm: partially invariant}), we arrive at nearly same bound on decay rate for both CMI and the trace norm CMI, up to fixed overheads and logarithmic correction.
\begin{prop}\label{prop:also trace norm CMI}
The trace norm CMI is bounded as
\begin{equation}
    I_1(A: C|B_{1}\cdots B_{\ell})_{\rho_\ell} =  O( \eta^{\ell} )
\end{equation}
for both states defined in Theorem~\ref{thm: bistochastic} and Theorem~\ref{thm: partially invariant}, with $\eta$ defined respectively.
\end{prop}

Unfortunately, we do not know to what extent the properties of CMI carry over, such as the connection to Markov states and recoverability. We leave further analysis on this function to future works.

\subsection{Decay of CMI for Y-shaped channels with a forgetful component}
\label{sec:forgetful}
Although we have shown the upper bound of CMI for bistochastic Y-shaped channels, they are only a small portion of the space of channels. Here, we show another suggestive result towards Question~\ref{question:main}. 

\begin{prop}
\label{prop:forgetful}
Consider a Y-shaped channel $\CN:C \rightarrow B_1 C$ such that there is a constant $0\le\eta <1$, a forgetful channel $F$ and another channel $\cN'$ such that 
\begin{equation}\label{eq:decF}
    \CN = (1-\eta) F + \eta \CN'\,,
\end{equation}
where $F[\rho] = \sigma $ maps any state to the same state. 
Then for any tripartite state $\rho_{A'B'C}$, the CMI contracts after applying $\CN$,
\begin{equation}
 I(A':C|B'B_1)_{\CE(\rho)} \le \eta  I(A':C|B')_\rho\,,
\end{equation}
\begin{equation}
 I_1(A':C|B'B_1)_{\CE(\rho)} \le \eta  I_1(A':C|B')_\rho\,. 
\end{equation}

Hence the resulting chain 
has exponentially decaying CMI and trace norm CMI for any tripartition 
\begin{align}
I(AB_1\ldots B_n: B_{n+m+1}\ldots B_\ell C_\ell|B_{n+1}\cdots B_{n+m})_{\rho_\ell} &=  O(\eta^{m})\,,\\
I_1(AB_1\ldots B_n: B_{n+m+1}\ldots B_\ell C_\ell|B_{n+1}\cdots B_{n+m})_{\rho_\ell} &=  O(\eta^{m})\,.
\end{align}

\end{prop}
  Remarkably, each application of $\CN$ contracts the CMI by $\eta$, from simple application of convexity of relatively entropy and monotonicity of CMI (proof in Sec.~\ref{proof:forgetful}). Intuitively the contraction ratio is $\eta$ because the forgetful channel necessarily removes any correlation with system $AB$. The step wise decay implies the decay of CMI for any tripartition with a long conditioned system $B$, i.e. the Y-shaped channel is applied many times. 

Note that any strictly positive Y-shaped channel $\CN(\rho)>0$ always has a decomposition~\eqref{eq:decF} with $F$ being the completely depolarizing channel. 
Moreover, for any Y-shaped channel $\CN'$ and any $0\leq\eta<1$, we can always construct a perturbed channel $\CN$ in Eq.~\eqref{eq:decF} satisfying $\|\CN-\CN'\|_{cb}\leq2(1-\eta)$.  

In the sense that any channel can be perturbed to have a forgetful component, Proposition~\ref{prop:forgetful} is true for most of quantum channels.  
We should note that the appropriate choice of `generic' MPDO certainly depends on the constraint of the problem at hand, for example if we restrict the number of Kraus operators to be small, then strict positivity would not hold generically.

A channel with forgetful component necessarily has trivial correctable algebra, but not vice versa. We can see this fact through an explicit example. 

{\bf Example:} The classical channel given by the following Kraus operators has trivial correctable algebra, but has no forgetful component. Input states $\ket{1}\bra{1}, \ket{2}\bra{2}, \ket{3}\bra{3}$ are mapped to $\frac{\ket{2}\bra{2}+ \ket{3}\bra{3}}{2}, \frac{\ket{1}\bra{1}+ \ket{3}\bra{3}}{2}, \frac{\ket{1}\bra{1}+ \ket{2}\bra{2}}{2}$, whose intersection is empty and thus cannot have a shared forgetful component.
\begin{align}
 K_1 = \frac{1}{\sqrt{2}} \left[
\begin{array}{ccc}
0  & 0 & 0 \\
1  & 0 & 0 \\
1  & 0 & 0 \\
\end{array}
\right], 
K_2 = \frac{1}{\sqrt{2}} \left[
\begin{array}{ccc}
0  & 1 & 0 \\
0  & 0 & 0 \\
0  & 1 & 0 \\
\end{array}
\right],  
K_3 = \frac{1}{\sqrt{2}} \left[
\begin{array}{ccc}
0  & 0 & 1 \\
0  & 0 & 1 \\
0  & 0 & 0 \\
\end{array}
\right].
\end{align}

\vspace{0.2 cm}

\subsection{A completely contractive DPI of CMI that implies exponential decay of CMI}\label{main DPI for CMI conjecture}
We have shown the channels we considered in Sec.~\ref{sec:bistochastic}, \ref{sec:forgetful} incur the exponential decay of CMI. For channels with a forgetful component (Proposition~\ref{prop:forgetful}), the exponential decay simply comes from a contraction occurring at each application of channel, and it is tempting to ask whether this is the general case. This motivates the following discussion on the data processing inequality.

First, we show that a single application of any channel with trivial correctable algebra always induces strict decay of CMI. 
\begin{prop}\label{prop:strictDPI} Let $\CE:C\to C'$ be a CPTP map which  has trivial correctable algebra. Consider a tripartite system  $A\ot B\ot C$. Then, for any state $\rho_{ABC}$ with $I(A:C|B)_\rho>0$, we have 
\begin{align}\label{eq:strictDPI}
    I(A:C'|B)_{\CE(\rho)} < I(A:C|B)_\rho\,.
\end{align}
\end{prop}
The proof is given in Sec.~\ref{proof:strictDPI}. By regarding $\CE$ as Y-shaped channel $\CN_{C_{i-1}\to B_iC_i}$, we see that $I(A:C|B)_{\rho_k}<I(A:C|B)_{\rho_{k-1}}$ holds for each $k$ (we abuse notation $BC$ for two different lengths). Note that having trivial correctable algebra provides only a sufficient condition to have a strict decay of the CMI. 
There is a Y-shaped channel with non-trivial correctable algebra which obeys Eq.~\eqref{eq:strictDPI} (see Ref.~\cite{KB20} for a complementary result). 

Unfortunately, recursively applying Eq.~\eqref{eq:strictDPI} may not imply the exponential decay of CMI of $\rho_\ell$ as the decay might get slower as system $B$ gets larger. We therefore propose the following conjecture to explicitly include a contraction ratio as a strong DPI constant.  
\begin{conj}[Completely Contractive DPI for the  CMI]\label{ccDPI}
For any channel $\CE: C \rightarrow C'$ with trivial correctable algebra $\CA(\CE)= \BC I$, there exists a constant $\eta < 1$ such that for any tripartite system $ABC$ and any state $\rho_{ABC}$, it holds that
\begin{equation}
 I(A:C'|B)_{\CE(\rho)} \le \eta I(A:C|B)_\rho.
\end{equation}
\end{conj}
Invoking standard monotonicity inequalities(Sec.~\ref{proof:onestep_decay}), the contraction accumulates and thus 
Conjecture~\ref{ccDPI} implies the desired exponential decay of CMI: 
\begin{prop}
\label{prop:DPItoCMI}
If Conjecture~\ref{ccDPI} is true, then for any $\rho_\ell(\CN)$ constructed by Y-shaped channel $\cN$ with trivial correctable algebra,
\begin{equation}
I(A':C'|B')_\rho = \CO( e^{-(\ln\eta)m} )
\end{equation}
for each tripartition $A'B'C'=ABC$ where $B'$ separates $A'$ from $C'$ with distance $m$. 
\end{prop}
Conjecture~\ref{ccDPI} holds for the channels with a forgetful component (Proposition~\ref{prop:forgetful}), where $\eta_{ABC}(\CE)$ does not approach 1 as $AB$ grows. 
In the classical case generally an unbounded auxiliary system can be reduced to have bounded dimension only depending on $C$, and for the strong DPI constants the auxiliary system can be ignored $\eta_{ABC} = \eta_{AC}= \eta_C$ (Theorem~\ref{thm:classicalSDPIc}). In the quantum case we do not know whether $AB$ can be reduced to a bounded dimension. Hence, the conjecture is open in general quantum systems and even showing its simpler variant for the mutual information (setting $B$ trivial) would be a breakthrough. Very recently the case when $A$ is classical and $B$ is trivial was proven with a contraction ratio analogous to the classical case~\cite{hirche2020contraction}.

A similar DPI can be proposed for the trace norm CMI.
\begin{conj}[completely contractive DPI for the trace norm CMI] \label{DPI trace norm}
For any channel $\CE: C \rightarrow C'$ with local contraction ratio
\begin{equation}
\eta_{1,C}:=\sup_{\rho_C,\rho'_C}\frac{\lVert\CE_C[\rho_{C}]-\CE_C[\rho'_{C}]\lVert_1 }{\lVert\rho_{C}-\rho'_{C}\lVert_1} <1
\end{equation}
There exists a global constant $\eta < 1$ such that for any tripartite system $ABC$ and any state $\rho_{ABC}$, it holds that
\begin{equation}
 I_1(A:C'|B)_{\CE(\rho)} \le \eta I_1(A:C|B)_\rho.
\end{equation}
\end{conj}
We do not know if extra constants or factor of dimension $d_C$ should be present between $\eta_{1,C}$ and $\eta$ like the crude bound in Theorem~\ref{thm: partially invariant}.
The trace norm seems more tractable than mutual information in the bipartite case.
\begin{prop}[Contraction of the trace norm mutual information]\label{prop:trace MI contraction}
For any channel $\CN: C \rightarrow C'$, it holds that 
\begin{align}
\lV \CN[\rho_{BC} -\rho_{B}\ot\rho_{C}] \rV_1 \le 4\eta_{1,C} d_C \lV \rho_{BC} -\rho_{B}\ot\rho_{C} \rV_1
\end{align} 
for any state $\rho_{BC}$.
\end{prop}
Here the contraction ratio is controlled by a bounded factor, where the dimension factor comes from bounding the completely bounded superoperator $1-1$ norm (Lemma~\ref{lem:cb to 1-1}).

\subsection{Sufficient conditions for exponential decay of classically-conditioned mutual information} 
In this section we consider the CMI of MPDO after performing the local measurement on the computational basis of $B$. Here we do not require the original map $\cN$ in Eq.~\eqref{def:family1} to be trace-preserving. 

\begin{thm}[Sufficient conditions for decay of CMI for measured MPDO]\label{mainthm}
For a B-measured MPDO as in Eq.~\eqref{eq:measured MPDO}, the following are implications in the order $1 \implies 2 \implies 3 \implies 4$:
\begin{enumerate}
\item (\textbf{Condition 1}) Each $\CM_s$ has $p$ Kraus operators $\{ E^s_1, \cdots, E^s_p\} $ such that $E^s_p$ is invertible, and \begin{equation*}
    \left\{\left(E^s_p\right)^{-1} E^s_1, \cdots, \left(E^s_p\right)^{-1}E^s_{p-1} \right\}
\end{equation*}
generate full matrix algebra $M_D(\mathbb{C})$ by addition and multiplication.

\item (\textbf{Condition 2}) There exists uniform coarse-graining length $\xi$ such that 
any sequence $\CM'_{t}:=\CM_{s_{k +l}}\cdots \CM_{s_k}$, where $t=(s_k,...,s_{k+l})$, is strictly positive map for all $ l\geq\xi$. In other words, 
 \begin{equation}
 {\rm det}(\CM'_t[\rho]) \ne 0\,,\quad \forall t, \forall\rho
 \end{equation}
 whenever $|t|\geq\xi$.

\item The mutual information of $A,C$ observing any outcome $b$ is bounded by a uniform decay rate $c_1>0$
\begin{equation}\label{eq:aawdecayMI}
I(A:C)_b \le c_0e^{-c_1\ell}.
\end{equation}

\item The measured conditional mutual information decays exponentially with the length $\ell$ of system $B$
\begin{equation}\label{decay_barCMI}
I(A:C|\bar{B}) \le c_0e^{-c_1\ell}.
\end{equation}

\end{enumerate}

\end{thm}
Condition 2 (strict positivity of all $\CM_s$) is unlikely a necessary condition for the decay of CMI, but it is convenient to state. We include coarse-graining in the statement as it is sometimes needed to ensure strict positivity of all $\CM_s$. For example, if the original Y-shaped CP map comes from tracing out a system much smaller than the bond dimension, then it may not be strictly positive.

A similar bound for the coarse-graining length $\xi$ for iteration of a single channel was given by the quantum Wielant inequality~\cite{5550282}: for every primitive channel $\CN$ (with Kraus rank $p$ and Hilbert space dimension $D$), $\xi= {D^2(D^2-p+1)}$ iteration guarantees full-rank output. Our Condition 1 is a multi-channel generalization: it is a sufficient condition for all possible sequences of $\CM_{s_\ell}\circ \cdots \circ\CM_{s_2} \circ\CM_{s_1}$ generated by CP maps $\{\CM_{1}, \CM_{2}, \cdots\CM_{n}\}$ to become strictly positive for coarse graining length $ l\geq\xi:=D^2-p+1$. See Sec.~\ref{sec:proof measured} for the proof. 

Although strict positivity (Condition 2) is NP-hard to check in general~\cite{gaubert2014checking},
we prove that Condition 1 can be verified in a polynomial time (see Sec.~\ref{proof:algorithm} for the details). 
\begin{prop}
\label{prop:algorithm}
Condition 1 can be checked in time polynomial in the bond dimension $D$ and the number of Kraus operators $p$.
\end{prop}
We numerically verify Condition 1 (generating full algebra) for Y-shaped channels generated by Haar-random Stinespting unitary (Fig.~\ref{fig:chain of Y-shaped}, bottom). The unitaries are the same random sample $U(D^2)$, and we have checked for computationally tractable local Hilbert space dimensions $D\le 8.$~\footnote{\cite{githubcode} is the repository of the jupyter code. The code checks condition 1 for arbitary MPDO given the Kraus operators.}. 
Both conditions imply Eq.~\eqref{eq:aawdecayMI}, which further implies our original goal, Eq.~\eqref{decay_barCMI}. 

\section{Proofs of main theorems}\label{sec:proof}
Here we present proofs of the main results by putting the key lemmas together, whose proofs we postpone in Sec.~\ref{sec:proof of lemmas}. 
\subsection{ Theorem~\ref{thm: bistochastic}: decay of CMI for bistochastic channel}
\label{sec: sketch bistochastic} 
The bistochastic assumption greatly simplifies the structure of state. The state can be decomposed into a uncorrelated state $\rho_{AB}\otimes \tau_C$ plus a deviation $G_{ABC}$ that is traceless on system $C$. At each application of channel in the generation of $\rho_\ell$, there are two key implications of being bistochastic: 1) the uncorrelated state $\rho_{AB}\otimes \tau_C$ is mapped to an uncorrelated state where system $C$ remains the maximally mixed state $\tau_C$, 2) the deviation $G_{ABC}$ contracts w.r.t the normalized Hilbert-Schimdt norm. The rest of the proofs are obtained by standard conversion between norms. The proofs of the lemmas are shown independently in Sec.~\ref{sec:remaining proof bistochastic}.

\begin{proof}
Let us take an operator basis containing identity which is orthogonal in the Hibert-Schmidt inner product. Then any operator on $C$ can be decomposed into the maximally mixed state $\tau_C$ and the traceless part $K_C$. 
After applying bistochastic Y-shaped channel $\CN:C_i\to B_{i+1}C_{i+1}$, the traceless part is mapped to
\begin{align}
\CN[K_C] &=O_B \otimes \tau_C + \sum_j O^j_B\otimes K^j_C \notag
\end{align}
where we abused the notation $O, O^j$ to denote general operators and $K,K^j$ to denote traceless operators. 
In the following we will use the curly brackets $\{O_B\otimes K_C\}:=\sum_jO_B^j\ot K^j_B$ as a shorthand notation for this type of linear combinations. Iteratively decomposing system C into the traceless and maximally mixed components after applying $\CN$, we can write down the structure of the state $\rho_{AB_1\cdots B_{\ell} C_{\ell}}$ explicitly. Starting with the maximally entangled input state $\sigma_{A\bar{A}} = \tau_{A}\otimes\tau_{\bar{A}}+ \{K_{A}\otimes K_{\bar{A}}\}$, we obtain that
\begin{align}
\rho_{AB_1C_1} &= \CN[\tau_{A\bar{A}}+\{K_{A}\otimes K_{\bar{A}}\}] = \tau_{AB_1}\otimes\tau_{C_1} + O_{AB_1}\otimes \tau_{C_1} + \{O_{AB_1} \otimes K_{C_1} \}\notag\\
&= \rho_{AB_1}\otimes \tau_{C_1} + \{ O_{AB_1} \otimes K_{C_1}\}, \\
\rho_{AB_1B_2C_2} &= \CN[\rho_{AB_1C_1} ] = \rho_{B_1}\otimes \tau_{B_2}\otimes \tau_{C_2} + O_{AB_1} \otimes O_{B_2} \otimes \tau_{C_2} + \{O_{AB_1} \otimes O_{B_2} \otimes K_{C_l}\} \notag\\
&=\rho_{AB_1B_2}\otimes\tau_{C_2} +\{ O_{AB_1B_2} \otimes K_{C_2}\},\\
\rho_{AB_1\cdots B_{\ell} C_{\ell}} & =  \tau_{AB_1} \otimes\cdots \tau_{C_\ell}\\
& +  O_{AB_1} \otimes\tau_{B_2}\cdots \tau_{C_\ell}\\
& + \cdots\\
& +  \{O_{AB_1} \otimes O_{B_2}\cdots \otimes K_{C_\ell}\}\\
& = \rho_{AB} \otimes \tau_{C_\ell}+ \{O_{AB}\otimes K_{C_\ell} \}\,,\label{eq:bOKk}
\end{align}
where we simply denote system $B_1\ldots B_\ell$ by $B$.  
The second term $\{O_{AB}\otimes K_{C_\ell} \}$ in Eq.~\eqref{eq:bOKk} is exponentially suppressed in the normalized Hilbert-Schmidt norm. To see this, we use the following lemma for bistochastic channels.
\begin{lem}[DPI for normalized HS-norm]\label{DPI 2norm}
For bistochastic channel $\CN_C[\tau_C] = \tau_{C'}$, the normalized Hilbert-Schimdt norm of operators $K_{EC}$ traceless on $C$ contracts. More precisely, if $Tr_C(K_{EC}) = 0$ 
then
\begin{equation}
\tr(d_{EC'}\CN_C[K]^2_{EC}) \le \eta  \tr(d_{EC}K^2_{EC})
\end{equation}
with the coefficient $\eta$ defined as
\begin{equation}
\eta := \limsup_{\rho_C \rightarrow \tau_C } \frac{D(\CN[\rho]\|\tau_{C'})}{D(\rho_C\|\tau_C)}.
\end{equation}
\end{lem}

Let $\BP_K$ be the orthogonal projection onto the traceless operator subspace on $C$, $G_{ABC}:=\{O_{AB}\ot K_C\}$ and $G_{A\bar{A}}:=\{K_A\ot K_{\bar A}\}$. Then, we obtain an expression 
\begin{align}
G_{ABC}= \BP_K\circ\cN\circ\BP_K\circ\cN...\BP_K\circ\cN[G_{A\bar{A}}]\,.
\end{align}  
By Lemma~\ref{DPI 2norm}, the normalized HS norm of $G_{ABC}$ contracts under $\BP_K\circ\CN$ since the projection does not increase the HS norm. We thus obtain 
\begin{align}
    d_{ABC}\tr(G_{ABC})^2 &= d_A d_{B_1...B_l} d_C \tr(\BP_K\circ\CN\circ\BP_K\circ\CN\circ...\CN(G_{A\bar{A}}))^2\\
    & \le d_A d_{B_1...B_l} d_C \tr(\CN\circ\BP_K\circ\CN\circ...\CN(G_{A\bar{A}}))^2\\
    & \le \eta\cdot d_AD_{B_{B_1...B_{\ell-1}}C} \tr(\BP_K\circ\CN\circ...\CN(G_{A\bar{A}}))^2\\
    &\vdots\\
    &\leq d_Ad_C\eta^l\tr(G_{A\bar A})^2\,.
\end{align}

By Cauchy-Schawrtz inequality, we can bound the trace norm of $G_{ABC}$ by the HS norm. 
\begin{equation}
\lV G_{ABC}\rV_1\le  d_{ABC}\tr(G_{ABC})^2 \le \eta^\ell d_A^2 \tr(G_{A\bar{A}})^2.
\end{equation}
Here we used $d_A=d_C$. 

Since $\|G_{ABC}\|_1$ is exponentially suppressed, $\rho_{ABC}$ is close to $\rho_{AB}\ot\tau_C$ for large $\ell$. The next step is to convert to CMIs, and clearly $\rho_{AB}\ot\tau_C$ has zero CMI. 
From the continuity of entropy, the difference of the CMI of two states are bounded as follows. 
\begin{lem}[Continuity of CMI]\label{cont of CMI}
If $\rho_{ABC} = \sigma_{ABC}+G_{ABC}$ and $d_{ABC}\tr(G_{ABC})^2 \le \epsilon \le 1/\mathrm{e}$,
then 
\begin{equation}
|I(A:C|B)_\rho -I(A:C|B)_\sigma |\le 4 \log(d_{ABC})\epsilon - 4 \log(\epsilon)\epsilon.
\end{equation}
\end{lem} 

By setting $\sigma_{ABC}=\rho_{AB}\ot\tau_C$ and taking large enough $\ell$ such that the deviation $G_{ABC}$ is small enough $\eta^\ell d_A^2 \tr(G_{A\bar{A}})^2\le 1/\mathrm{e}$, we obtain the exponential decay of the CMI:
\begin{align}
I(A:C|B)_\rho &\le 4(\log(d_{B_1})\ell+2\log(d_A))\eta^\ell d_A^2 \tr(G_{A\bar{A}})^2 \notag\\
&\hspace{3.5cm}- 4 \left(\ell \log(\eta) +\log(d^2_A\tr(G_{A\bar{A}})^2)\right)\eta^\ell d_A^2 \tr(G_{A\bar{A}})^2\\
&\le O(\ell\eta^{\ell}).
\end{align}

Lastly, we alternatively characterize $\eta$ by a spectral property:
\begin{lem}\label{lem:second singular}
For a bistochastic channel $\CE$, the second largest singular value (which coincides with eigenvalue due to being self-adjoint) of the Petz-recovery channel composed with the channel $\CP_{\tau, \CE}\circ\CE$ is exactly the contraction ratio $\eta$.
    \begin{equation}
        \eta = \lambda_2(\CP_{\tau, \CE}\circ\CE) = \limsup_{\rho \rightarrow \tau } \frac{D(\CE[\rho]\|\tau')}{D(\rho\|\tau)}\,.
    \end{equation}
\end{lem} 
Therefore $\eta<1$ if and only if the largest singular value (or equivalently the eigenvalue) $\lambda_1=1$ is unique. By Proposition~\ref{invar_cor}, the unit eigenvalue subspace of $\lambda_1$ being one-dimensional is equivalent to the correctable algebra being trivial.
\begin{equation}
    \CA(\CE) = \BC I \iff   \lambda_2(\CP_{\tau, \CE}\circ\CE) = \eta <1\,.
\end{equation}
This completes the second statement of the proof.
\end{proof}

\subsection{Theorem~\ref{thm: partially invariant}: decay of trace norm CMI for partially invariant channel}
\label{sec: sketch partially invariant}
The proof structure is analogous to the previous Section~\ref{sec: sketch bistochastic}, with the key distinction being the norm and the associated tools. Under the trace norm we have results for a slightly general family of channels satisfying $\CN[\nu_C] = \sigma_{B}\otimes \nu_{C}$. Though, unlike the HS norm, the trace norm suffers from tensoring with an auxiliary system by a bounded factor of dimension of system $C$ (Lemma~\ref{lem:cb to 1-1}). The contraction ratio obtained this way is likely not the most stringent, but at least it implies strict contraction of the CMI when the channel is sufficiently noisy. The proofs for the key lemmas are shown in Sec.~\ref{sec:remaining proof partially invariant}.
\begin{proof}

First, we decompose the state into a product state $\rho_{AB}\otimes \nu_C$ plus some deviation $G_{ABC}$ traceless on system $C$.
Suppose 
\begin{align}
\CN[K_C] &= \BP_\nu \circ\CN[K_C]+ (\id- \BP_\nu)\circ \CN[K_C] = O_B \otimes \nu_C + \{O_B\otimes K_C\},
\end{align}
where $\BP_\nu:=\nu_C Tr_C[\cdot]$ is the forgetful channel on subsystem C. Again we used $K_C$ to denote an operator traceless on system $C$, but we should note the difference from Sec.~\ref{sec: sketch bistochastic} that the projection $(\id-\BP_\nu)$ is not an orthogonal projection under HS norm. Writing the maximally entangled state as $\sigma_{A\bar{A}} = \tau_A \otimes\nu_{\bar{A}} + \{O_A\otimes K_{\bar{A}}\}$, we obtain that 
\begin{align}
\rho_{AB_1C_1} &= \CN[\tau_A\otimes \nu_{\bar{A}}+\{O_A\otimes K_{\bar{A}}\}] =\tau_A\otimes \sigma_{B_1}\otimes\nu_{C_1} + O_{AB_1} \otimes\nu_{C_1} +\{ O_{AB_1} \otimes K_{C_1}\}\notag\\
&= \rho_{AB_1}\otimes\nu_{C_1} + \{O_{AB_1} \otimes K_{C_1}\},
\\
\rho_{AB_1B_2C_1} &= \CN[\rho_{AB_1C_1} ]\\ 
&=\rho_{AB_1} \otimes\sigma_{B_2}\otimes\nu_{C_2} + O_{AB_1} \otimes O_{B_2}\otimes \nu_{C_2} + \{O_{AB_1} \otimes O_{B_2} \otimes K_{C_2}\} \notag\\
&=\rho_{AB_1B_2}\otimes\nu_{C_2} + \{O_{AB_1B_2} \otimes K_{C_2} \},\\
\rho_{AB_1\cdots B_{\ell} C_{\ell}} & =  \tau_A\otimes\sigma_{B_1} \otimes\sigma_{B_2}\cdots \nu_{C_\ell}\notag\\
& +  O_{AB_1}\otimes \sigma_{B_2}\cdots \nu_{C_\ell}\notag\\
& +\cdots\notag \\
& +  \{O_{AB_1} \otimes O_{B_2}\cdots \otimes K_{C_\ell}\}\notag\\
& = \rho_{AB} \otimes \nu_{C_\ell}+\{ O_{AB}\otimes K_{C_\ell} \}\\
&= \rho_{AB}\otimes \nu_{C_\ell}+ (\id-\BP_\nu)\circ\CN...(\id-\BP_\nu)\circ\CN(G_{A\bar{A}}).
\end{align}
Let $G_{ABC} := (\id-\BP_\nu)\circ\cN...\circ(\id-\BP_\nu)\circ\cN(G_{A\bar{A}})$, which is the output of iteration of the map $(\id-\BP_\nu)\circ\cN\circ (\id-\BP_\nu)$, where we multiply a copy of $(\id-\BP_\nu)$ as its square is equal to itself. The contraction ratio can be obtained by 
\begin{equation}\label{eq:tr1gabc}
\lV G_{ABC}\rV_1 \le \lV(\id-\BP_\nu)\circ \CN\circ (\id-\BP_\nu)\rV_{cb}^\ell \lV G_{A\bar{A}}\rV_1,
\end{equation}
where the completely bounded norm for a super-operator $\phi :\CM_d \rightarrow \CM_{d'}$ is defined by\footnote{This is 1-1 super-operator norm with arbitrarily large ancilla.} 
  \begin{align}
  \lV\phi \rV_{cb}:=\sup_{k, X} \frac{\lV( \phi\otimes id_k )[X]\rV_1}{\lV X \rV_1}.
  \end{align}
Unfortunately, the completely bounded norm of the difference of channels is not less than one in general. However, we do obtain a meaningful contraction bound for $\eta_{1,C}$ small enough. 
\begin{lem}
\label{cb eta}
The completely bounded norm of the difference between channels is bounded by 
\begin{equation}
\lV(\id-\BP_\nu)\circ \CN\circ (\id-\BP_\nu)\rV_{cb} \le 2\lV  \CN - \CN\circ\BP_\nu\rV_{cb} \le 16 d_C \eta_{1,C} ,
\end{equation}
where $\eta_{1,C}$ is the trace norm contraction ratio
\begin{equation}
\eta_{1,C}:=\sup_{\rho_C,\rho'_C}\frac{\lVert\CN[\rho_{C}]-\CN[\rho'_{C}]\lVert_1 }{\lVert\rho_{C}-\rho'_{C}\lVert_1}.
\end{equation}

\end{lem}
The constant factor $16$ is crude, and the factor of dimension $d_C$ is due to converting 1-1 superoperator norm to the diamond norm.
From this lemma Eq.~\eqref{eq:tr1gabc} reduces to 
\begin{equation}
    \lV G_{ABC}\rV_1 \le \left(16 d_C \eta_{1,C} \right)^\ell \lV G_{A\bar{A}}\rV_1. 
\end{equation}
By the triangle inequality we show that $\rho_{ABC}$ becomes close to a product state in trace norm
\begin{equation}
\left\|\rho_{ABC} - \rho_{AB}\otimes \nu_C\right\|_1 \leq O\left((16\eta_C d_C)^\ell\right).
\end{equation}
We conclude the proof by calling the continuity of CMI (Lemma~\ref{cont of CMI}) again.
\end{proof}

\subsection{Thoerem~\ref{mainthm}: decay of $B$-measured CMI in MPDO}\label{sec:proof measured}

First, $3 \implies 4 $ is immediate from that the measured CMI is equal to the expectation of mutual information between $A$ and $C$ over measurement outcomes $b$ in $\bar{B}$, 
\begin{align}\label{eq:expectation over MI}
I(A:C|\bar{B})_\rho =  \sum_{b}p_{\bar B}(b)I(A:C)_b,
\end{align}
where $I(A:C)_b$ is the mutual information of state $\rho_{AC|\,{\bar B}=b}\propto \tr_{\bar B}(|b\>\<b|\rho_{ABC})$, defined on the Hilbert space of system $A$ and $C$ only, as we can individually discuss each classical outcome $\bar{B}=b$.

We show the remaining two indications $1\implies2$ and $2\implies3$ in the following. 

\subsubsection{$2 \implies 3$:  uniformly bounding every sequence quantum operations by contraction in Hilbert's projective metric }
\begin{proof}
We show Condition 2 implies the mutual information $I(A:C)_b$ decays exponentially with a uniform decay rate $c$ for every outcome $b$. As showed in the above, this implies the decay of CMI without touching the probability distribution $p(b)$.  

Recall that we have a set of CP-maps $\cM_{s_k}:C\to C$ for each outcome $s_k$.  These maps $\CM_{s_{k}}$(or the coarse-grained maps $\CM'_t:=\CM_{s_{k+\xi}}\circ\ldots\circ\CM_{s_{k}}$) are all contractions in 
the \textit{Hilbert's projective metric} $h(a,b)$. A self-contained introduction and supporting theorems are included in Section~\ref{hilbert}. The contraction ratio associated to a CP-map $\CM$ 
is given by
\begin{align}\label{contraction ratio}
\eta_\CM := \sup \limits_{a,b \in S_+} \frac{h(\CM(a),\CM(b))}{h(a,b)} \le 1,
\end{align}
where the supremum is taken over $S_+$, the set of unnormalized positive semi-definite  operators.  Condition 2 guarantees that $\{\CM'_t\}$ has a strict contraction ratio at Eq.~\eqref{contraction ratio}. After repeatedly applying $\CM'_t$, any state is mapped to a fixed state (up to rescaling) with exponentially small error. Accounting the normalization and the conversion to trace norm,  we obtain the following lemma.
\begin{lem}\label{lem:chpmttn}
\label{lem: hilbert to trace}
If all $\{\CM_{1}, \CM_{2}, \ldots,\CM_{n}\}$ are CP self-maps that map any state to a full rank state, then for arbitrary sequence of $b = s_\ell,\ldots, s_1 \in \{1,\ldots, n\}^\ell$ and all $\rho_1, \rho_2$
\begin{align}
\left\| \frac{\CM_b[\rho_1]}{\tr(\CM_b[\rho_1])} - \frac{\CM_b[\rho_2]}{\tr(\CM_b[\rho_2])} \right\|_1 = O(e^{-c \ell}) ,
\end{align}
where the exponent $c$ is independent of $b$. 
\end{lem}

This statement implies the bipartite state $\rho_{AC,b}$ being close to a product state in the trace norm.
\begin{lem}
\label{forget to MI}
Suppose the CP map $\CM_b: C' \rightarrow C$ is contractive in the sense that for all $\rho_1, \rho_2$,
\begin{equation}
    \left\lV \frac{\CM_b[\rho_1]}{\tr(\CM_b[\rho_1])} - \frac{\CM_b[\rho_2]}{\tr(\CM_b[\rho_2])} \right\rV_1 \le \epsilon.
\end{equation}
Then the state $\rho_{AC} := \CM_b[\sigma_{AC'}]/\tr(\CM_b[\sigma_A])$ is close to the product state in the trace distance
\begin{align}
2T_b := \lV \rho_{AC,b} - \rho_{A,b}\otimes\rho_{C,b} \rV_1 \le 4d_{C'}\lV \sigma_{C'}^{-1}\rV \epsilon.
\end{align}
\end{lem}
Plugging $\sigma_{C'} = \tau_{C'}$, and by the Alicki–Fannes–Winter inequality~\cite{winter2015tight} (see also \cite[Theorem~11.10.3]{wilde2011classical}), 
we convert the $O(e^{-c\ell})$ bound on the trace norm to the mutual information with some constant overhead.
\end{proof}

\subsubsection{$1 \implies 2$: uniformly bounding coarse graining length by generalizing quantum Wielandt's inequality to family of CP self-maps}

Recall our goal is to show that there exists a coarse-graining length $\xi$ such that for any sequence $t=(s_{\ell},...,s_{1})$, the coarse grained CP map $\CM'_{t}:=\CM_{s_{l}}\cdots \CM_{s_1}$ is strictly positive for all $ l\geq\xi$. First, we convert strict positivity to having full Kraus rank.
\begin{lem}\label{fulldimkraus}
If ${\rm span}\{K_1, \cdots, K_p \} = Mat(D,D)$ then $\CN[\rho]: = \sum_i K_i \rho K_i^\dagger > 0, \forall \rho$.
\end{lem}

Our strategy to a uniform bound on $\xi$ now relies on demanding for each $s_k$,  $\CM_{s_k}[\cdot] = \sum_i E_i^{s_k} \cdot E_i^{s_k\dagger}$ must increase the dimension of span of Kraus operators untill reaching full rank. 
\begin{lem}[condition 1 implies the increment of span of Kraus operators]\label{increase span}
Consider a CP self-map $\CM$ with $p>2$ Kraus operators $\CM\sim \{ E_1,\cdots, E_p\} $ such that $E_p$ is invertible. Suppose that $\left((E_p)^{-1} E_1,\cdots,(E_p)^{-1}  E_{p-1} \right)$ generate full matrix algebra $Mat_\BC(D,D)$ by addition and multiplication. Then for any set of Kraus operators $\{T_j\}$ containing an invertible element and whose Kraus rank is not full, the dimension of span must increase after applying $\CM\sim \{ E_1,\cdots, E_p\} $
\begin{equation}
{\rm dim} ({\rm span}\{E_i T_j\}) > {\rm dim} ({\rm span}\{T_j\}). 
\end{equation}
\end{lem}

We are now ready to prove $1\implies 2$:
\begin{proof}

For any sequence $t=(s_{\ell},...,s_{1})$, consider the linear span of Kraus operators of the product $\CM_{s_k}\cdots \CM_{s_2}\CM_{s_1}$ ($1\le k\le \ell$)
$$
\CS_{k}: = {\rm span}\{E^{s_k}_{q_k}\cdots E^{s_2}_{q_2}E^{s_1}_{q_1}\},
$$
where $\{E^{s_1}_{q_1}\}_{q_1}$ are the Kraus operators of the CP map $\CM_{s_1}$. We have shorthanded the above as $\CS_k$, while keeping in mind its dependence on the sequence $s_k \cdots s_1$. 
Applying $\CM_{s_{k+1}}$ amounts to left-multiplying its Kraus operators, giving the new span of Kraus operators:
$$
S_{k+1}= \{ E^{s_{k+1}}_1\CS_{k}, \cdots ,E^{s_{k+1}}_{p-1}\CS_{k}\}.
$$
By Lemma~\ref{increase span}, multiplying $\CM_{s_{k+1}}$ must increase the dimension of span of Kraus operators. Therefore, after constant $\ell \ge \xi:=D^2-p+1$ steps, the product CP map $\CM_{s_\ell} \cdots \CM_{s_2} \CM_{s_1}$ must reach full Kraus rank and hence becomes strictly positive (Proposition~\ref{tweakburnside}). Importantly, the same coarse graining length $\xi$ works for all sequences $(s_{\xi},\cdots, s_1)$, which establishes Condition 2. Note that since $\xi$ is finite, the contraction ratios of $\CM_t$ has a global bound being strictly less than one. 

\end{proof}

\section{Proofs for remaining propositions and lemmas}\label{sec:proof of lemmas}

Here we compile the remaining proofs and lemmas, categorized by the theorem they are supporting:  Theorem~\ref{thm: bistochastic} at Section~\ref{sec:remaining proof bistochastic}; Theorem~\ref{thm: partially invariant} at Section~\ref{sec:remaining proof partially invariant};  Theorem~\ref{mainthm} at Section~\ref{sec:remaining proof measured}.
\subsection{Remaining Proof for Theorem~\ref{thm: bistochastic}}\label{sec:remaining proof bistochastic}
\subsubsection{Proof of Lemma~\ref{DPI 2norm}, DPI for 2-norm}
\begin{proof}
By definition, the contraction can be expressed by the completely bounded superoperator 2-2 norm of the channel. 
\begin{align}
  \tr(d_{EC'}\CN_C[K]^2_{EC}) &\le \lV\CI\otimes(\CN\circ\BP_K)\rV_{2-2}^2\frac{d_{EC'}}{d_{EC}}\tr(d_{EC}K^2_{EC})\\
  &= \lV\CI\otimes(\CN\circ\BP_K)\rV_{2-2}^2\frac{d_{C'}}{d_{C}}\tr(d_{EC}K^2_{EC}) .
\end{align}
The superoperator $2-2$ norm is equal to the largest-singular-value of the map (with input restricted to be traceless on system $C$) w.r.t. the 2-norm of operators. Tensoring with identity does not change the leading singular value, and hence the completely bounded $2-2$ norm can be evaluated without auxiliary system  as
\begin{align}
\lV\CI\otimes(\CN\circ\BP_K)\rV_{2-2}^2\frac{d_{C'}}{d_{C}} &=\lV\CN\circ\BP_K\rV_{2-2}^2\frac{d_{C'}}{d_{C}} \\
&=\sup_{K_C, \tr(K_C)=0} \frac{\tr(\CN_C[K]^2_{C}) }{\tr(K^2_{C})}\cdot\frac{d_{C'}}{d_{C}} \\
&= \lim_{\rho_C \rightarrow \tau_C }\sup \frac{D(\CN[\rho]\|\tau_{C'})}{D(\rho_C\|\tau_C)},\label{eq:22norm relative}
\end{align}

In the last equality we convert the 2-2 norm to the perturbation of relative entropy
\begin{align}
D(\rho\|\tau) = -S(\rho) + \log(d) =  \frac{1}{2}\tr(dK^2) + O(K^3).
\end{align}
The $O(K^3)$ and higher order terms vanish in the limit because both spaces $C$ and $C'$ have finite dimension, which guarantee the operator norm vanishes $\lV K\rV\rightarrow 0$ in the limit of $\rho_C\to\tau_C$.
\end{proof}

\subsubsection{Proof of Proposition~\ref{cont of CMI}, the continuity of CMI}
The proof is based on the continuity of entropy w.r.t. the normalized H-S norm

\begin{lem}[Continuity of entropy]\label{cont of entropy}
\begin{align}
|S(\rho) - S(\sigma)| &\le \lV\rho-\sigma\rV_1 \log(d) - \lV\rho-\sigma\rV_1 \log(\lV\rho-\sigma\rV_1) \\
&\le d\lV \rho-\sigma\rV_2^2 \left(\log(d)- \log(d\lV \rho-\sigma\rV_2^2) \right).
\end{align}
\end{lem}
Taking partial trace does not increase the normalized HS norm. 
\begin{prop}\label{partial trace h-s norm}
\begin{align}
d_{A}\tr(G_A^2)\le d_{AB}\tr(G_{AB}^2).
\end{align}
\end{prop}
\begin{proof}
Decompose othogonally $G_{AB} = \tau_A\otimes G_B + \{K_{A}\otimes O_B\}$ so that
\begin{align}
d_{AB}\tr(G_{AB}^2) =    d_{AB}\tr(\tau_A^2 \otimes G_B^2 + (K_A\otimes O_B)^2) \ge d_{AB}\tr(\tau_A^2\otimes G_B^2) = d_{B}\tr(G_B^2).
\end{align}

\end{proof}
We can now prove the continuity of CMI, Proposition~\ref{cont of CMI}.
\begin{proof}
We can expand the CMI as 
\begin{align}
&I(A:C|B)_\rho -I(A:C|B)_\sigma  \\
&=  (S_{\rho}(ABC)- S_\sigma(ABC))- (S_{\rho}(AB) - S_\sigma(AB)) - (S_{\rho}(BC) - S_\sigma(BC))+ (S_{\rho}(B) - S_\sigma(B)).
\end{align} 

By the continuity of entropy (Lemma~\ref{cont of entropy}),  it holds that
\begin{align}
|(S_{\rho}(AB) - S_\sigma(AB)) | &\le d_{AB}\lV G_{AB}\rV_2^2 \left(\log(d)- \log(d_{AB}\lV G_{AB}\rV_2^2) \right)\\
 &\le d_{ABC}\lV G_{ABC}\rV_2^2 \left(\log(d)- \log(d_{ABC}\lV G_{ABC}\rV_2^2) \right)\\
 &\le \log(d_{ABC})\epsilon - \log(\epsilon)\epsilon.
\end{align}
In the second inequality we used Proposition~\ref{partial trace h-s norm} and that function $-x\log(x)$ is monotonically increasing at $x \le 1/e$. The same argument holds for all four terms. 
\end{proof}

\subsubsection{Proof of Lemma~\ref{lem:second singular}, spectral characterization of the contraction ratio}
\begin{proof}
We expand the Petz-recovered map
\begin{align}
    \CP_{\tau, \CE}\circ\CE[\cdot] &= \tau^{1/2} \CE^\dagger [\CE[\tau]^{-1/2}\CE[\cdot] \CE[\tau]^{-1/2}] \tau^{1/2}\\
    &= \frac{d'}{d} \CE^\dagger [\CE[\cdot]]
\end{align}
where $\cN$ being bistochastic substantially simplifies the expression, making it self-adjoint and positive. Hence, the spectrum is positive and the second eigenvalue/singular value, expressed by removing the trace coincides with the relative entropy characterization Eq.\eqref{eq:22norm relative}
\begin{equation}
  \lambda_2(\CP_{\tau, \CE}\CE) = \frac{d'}{d}  \lV \CE \BP_K \rV^2_{2-2} = \limsup_{\rho_C \rightarrow \tau_C } \frac{D(\CE[\rho_C]\|\tau_{C'})}{D(\rho_C\|\tau_C)}.
\end{equation}
\end{proof}

\subsection{Proofs for propositions around Conjecture~\ref{ccDPI}}
\subsubsection{Proof of Proposition~\ref{prop:strictDPI}}\label{proof:strictDPI}
We prove a slightly more general lemma that includes the case when the correctable algebra is not trivial, which immediately converts to Proposition~\ref{prop:strictDPI}. Some necessary background are at Appendix~\ref{sec:Petz and correctable}. 
\begin{lem}\label{lem:equalMI}
For all tripartite state $\rho_{ABC}$ and channel acting on system $C$ only $\CE: C\rightarrow D$
\begin{equation}
I(A:BC)_{\rho} = I(A:BD)_{\CE[\rho]} \iff I(A:B\gamma)_\rho = I(A:BC)_\rho,
\end{equation}
where $\gamma: =\CA(\CE)$ is the correctable algebra of $\CE$, and $I(A:B\gamma)_\rho:=I(A:BC)_{E_\gamma(\rho)}$ with $E_\gamma:C\to \gamma$ the conditional expectation (see e.g., ~\cite{TakesakiI}) onto the subalgebra $\gamma$ \footnote{Sometimes $E_\gamma$ is called the restriction to subalgebra (see e.g.~\cite{Chen_2020} ). If $\gamma$ is trivial then this is taking partial trace over $C$.}.
\end{lem}
In words, if a channel does not decrease the mutual information of the state, then the correlation with system $A$ must be perfectly stored in the correctable algebra $\CB(\CH_B)\otimes \gamma = \CA(\CI_B\ot \CE_C)$. Operationally, the LHS implies that from reduced state $E_\gamma( \rho_{ABC})$ there is some recovery channel $R: B\gamma \rightarrow BC$ that recovers the full state $\rho_{ABC}$.
\begin{proof} 
We start with the recoverability theorem(see, e.g. ~\cite[Corollary~12.5.1]{wilde2011classical})
\begin{align}
I(A:BC)_\rho = I(A:BD)_{\CE[\rho]}& \iff D(\rho_{ABC}\|\rho_A\otimes\rho_{BC}) =  D(\CE_C[ \rho_{ABC}]\|\rho_A\otimes\CE_C[\rho_{BC}])\\
&\implies \CP_{\rho_A\otimes\rho_{BC},\CE}\circ \CE[\rho_{ABC}] =\rho_{ABC} \end{align}
where we are using the Petz map with reference state $\rho_A\otimes\rho_{BC}$, expanded explicitly as follows
\begin{align}
 \CP_{\rho_A\otimes\rho_{BC},\CE}[\rho] &= \sqrt{\rho_A\ot\rho_{BC}} \CE^\dagger [ \sqrt{\CE[\rho_A\ot\rho_{BC}]}^{-1} \rho    \sqrt{\CE[\rho_A\ot\rho_{BC}]}^{-1}  ]\sqrt{\rho_A\ot\rho_{BC}}\\
& =  \sqrt{\rho_{BC}} \CE^\dagger [ \sqrt{\CE[\rho_{BC}]}^{-1} \rho    \sqrt{\CE[\rho_{BC}]}^{-1}  ]\sqrt{\rho_{BC}} = \CI_A \otimes  \CP_{\rho_{BC},\CE}[\rho]
\end{align}
The exact form of the Petz map does not mean so much for us here, and all we care is the recovery map $\CP_{\rho_A\ot\rho_{BC},\CE} $ only acts on subsystem $BC$ due to factorization of the reference state $\rho_A\otimes\rho_{BC}$. Therefore we have the invariant subspace equation for $\CZ_{BC} = \CP_{\rho_{BC},\CE}\circ\CE$ and $\rho_{ABC}$, 
\begin{equation}
\CI_A\otimes \CZ_{BC}[\rho_{ABC}] = \rho_{ABC}.
\end{equation}
Suppose the fixed point algebra has factors $\alpha$,
\begin{align}
 S_{\CZ_{BC}^\dagger}& = \bigoplus_\alpha B(\CH_\alpha)\otimes I_{\bar{\alpha}}\subset B(\CH_{BC})\,.
\end{align}
where the $\alpha$s are dependent on $\CZ_{BC}$, and $\bar{\alpha}$ labels the factors of the commutant of $S_{\CZ^\dagger}$ and each $\bar{\alpha}$ has one-to-one correspondence with $\alpha$.
Then since $S_{\CI_A\ot \CZ_{BC}^\dagger} = \CB(\CH_{A})\otimes S_{\CZ_{BC}^\dagger}$, by Theorem~\ref{invar}, $\rho_{ABC}$ must have the Markovian structure characterized by $\alpha$s as follows: 
\begin{equation}
\rho_{ABC} = \bigoplus_\alpha p_\alpha \rho_{A\alpha} \otimes \sigma_{\bar{\alpha}}\,.
\end{equation}

To know how $\alpha$s are embedding in system $BC$, certainly $S_{\CZ_{BC}^\dagger}$ are by definition correctable, i.e. subalgebra of the correctable algebra of $\CI_B \otimes \CE$. Hence by Proposition~\ref{correctabletensor}, denoting the factors of the correctable algebra $\gamma=\CA(\CE)$ by $\beta$,
\begin{equation}
S_{\CZ_{BC}^\dagger} \subset \bigoplus_\beta B(\CH_{B\beta})\otimes I_{\bar{\beta}}
\end{equation}
Taking the dual, this is saying restricting to subalgebra $B\gamma$ keeps the components $\rho_{A\alpha}$ intact, and thus does not change the mutual information between $BC$ and $A$, 
\begin{equation}
 I(A:B\gamma)_\rho = I(A:BC),
\end{equation}
completing the proof.
\end{proof}
We can now prove Proposition~\ref{prop:strictDPI}, where $\CE$ has trivial correctable algebra. Rewriting Lemma~\ref{lem:equalMI},
\begin{align}
I(A:C|B)_{\rho} = I(A:BC)_{\rho} - I(A:B)_\rho &> I(A:BD)_{\CE[\rho]} - I(A:B)_\rho = I(A:C'|B)_{\CE[\rho]}\\
&\iff I(A:B)_\rho < I(A:BC)_\rho \\
& \iff I(A:C|B)_\rho > 0
\end{align}
,where the strict inequalities are from taking the negation of equality. 
\subsubsection{Proof of Proposition~\ref{prop:DPItoCMI}}\label{proof:onestep_decay}
The proof uses standard manipulation of CMI:
\begin{align}
I(A:C_3|B_1B_2B_3) &= I(A: B_1B_2B_3C_3) - I(A:B_1B_2B_3)\\
&\le I(A: B_1B_2B_3C_3) - I(A:B_1B_2)\label{eq:C to B}\\
& \le \eta (I(A:B_1B_2C_2) - I(A:B_1B_2)) \label{eq:lose C}\\ 
&= \eta I(A:C_2|B_1B_2) .
\end{align}
where in the first inequality we used the monotonicity of the mutual information. The second inequality uses Conjecture~\ref{ccDPI}
for $\CE = \CN: C_2 \rightarrow B_3 C_3 , B = B_1B_2, C = C_2$. For each tripartition separated by $m$ sites, 
we obtain the exponential decay of CMI by iterating this argument $m$ times. The proof is identical for the trace norm CMI.

\subsubsection{Proof of Proposition~\ref{prop:forgetful}}\label{proof:forgetful}
\begin{proof}
It suffices to show $\CN$ satisfies the DPI conjecture~\ref{question:main}, using joint convexity of relative entropy 
 \begin{align}
 I(A:BC')_{\CN_C[\rho]} &= D\left( \CN[\rho_{ABC}]\left\|\rho_A\otimes\CN[\rho_{BC}] \right.\right) \notag\\
 &= D\left((1-\eta) F[\rho_{ABC}]+ \eta\CN'[\rho_{ABC}]\left\|(1-\eta) \rho_A\otimes F[\rho_{BC}]+ \eta\rho_A\CN'[\rho_{BC}]\right.\right)\notag\\ 
 &\le   (1-\eta) D(F[\rho_{ABC}]\|\rho_A\otimes F[\rho_{BC}]) + \eta D(\CN'[\rho_{ABC}]\|\rho_A\otimes\CN'[\rho_{BC}]) \notag\\
&=  (1-\eta) D(\rho_{AB}\ot\sigma\|\rho_A\otimes\rho_{B}\ot\sigma)+ \eta D(\CN'[\rho_{ABC}]\|\rho_A\otimes\CN'[\rho_{BC}])\notag\\
&\le (1-\eta) D(\rho_{AB}\|\rho_A\otimes\rho_{B})+ \eta D(\rho_{ABC}\|\rho_A\otimes\rho_{BC})\\
& = I(A:B)+  \eta \left( I(A:BC)- I(A:B) \right) \,.
 \end{align}

  In the first inequality we used the joint convexity of relative entropy, and in the second inequality we used the DPI under $\CN'$. We conclude the proof by moving $I(A:B)$ back to the LHS. The joint convexity holds for the trace-norm CMI as well and the lines are identical.
  Now we follow the lines in Sec.~\ref{proof:onestep_decay} to get step-wise decay of CMI and then the decay of CMI for each tripartition.
 \end{proof}

\subsubsection{Proof of Proposition~\ref{prop:trace MI contraction}}
\begin{proof}
We control the completely bounded superoperator $1-1$ norm of $\CI_B\otimes \CN$ by using the following lemma:
\begin{lem}[{\cite[Section~3.11]{paulsen_2003}}]\label{lem:cb to 1-1}
For arbitrary map $\phi :\CM_d \rightarrow \CM_{d'}$, the completely bounded superoperator $1-1$ norm is at most the dimension $d$ times the $1-1$ norm.
\begin{equation}
    \lV\phi \rV_{cb}:=\sup_{k, X} \frac{\lV( \phi\otimes id_k )[X]\rV_1}{\lV X \rV_1}  \le d \sup_{X} \frac{\lV\phi[X]\rV_1}{\lV X \rV_1}  = d \lV \phi\rV_{1-1}.
\end{equation}
\end{lem}
Then the proof of Prop.~\ref{prop:trace MI contraction} simply follows as
\begin{align}
\lV \CN[\rho_{BC} - \rho_{B}\otimes\rho_{C}] \rV_1 
& = \Big\lV(\CN - \CN\circ (\rho_CTr_C))[\rho_{BC}-\rho_{B}\otimes\rho_{C}]  \Big\rV_1\\
&\le 4\eta d \lV \rho_{BC} -\rho_{B}\otimes\rho_{C} \rV_1,
\end{align} 
where in the first line we insert a vanishing term $\rho_CTr_C[\rho_{BC}] - \rho_B\ot \rho_C =0$, and in the second line we bounded the completely bounded trace norm using Lemma~\ref{lem:cb to 1-1}.
\end{proof}

\subsection{Remaining proof for Theorem~\ref{thm: partially invariant}}\label{sec:remaining proof partially invariant}

\subsubsection{Proof of Lemma~\ref{cb eta}}

\begin{proof}
\begin{align}
    \lV(\id-\BP_\nu)\circ \CN\circ (\id-\BP_\nu)\rV_{cb} 
    &\le 2\lV  \CN - \CN\circ\BP_\nu\rV_{cb} \\
    &\le 2d_C\lV  \CN - \CN\circ\BP_\nu\rV_{1-1}\\ 
    &\le 8d_C \sup_{\rho \in S^+}\lV (\CN - \CN\circ \BP_\nu )[\rho]\rV_1\\
    &\le 8d_C\sup_{\rho \in S^+}\lV \CN [\rho -\BP_\nu \rho ]\rV_1\\
    &\le 16d_C \sup_{\rho,\rho'}\frac{\lVert\CN[\rho]-\CN[\rho']\lVert_1 }{\lVert\rho-\rho'\lVert_1}\\
    &\le 16 d_C \eta_{1,C} 
\end{align}

We used $\lV\id -\BP_\nu \rV_{cb} \le 2$ in the first inequality. The second inequality follows from Lemma~\ref{lem:cb to 1-1}. In the third inequality, we convert the optimization over operator $\lV X\rV_1$ into positive $\rho$ by a factor of 4. This chain of inequalities is largely along the lines of \cite[Theorem~45]{James2015QuantumMC}. 
\end{proof}
\subsubsection{Proof of Proposition~\ref{prop:also trace norm CMI}}
We get the decay of trace norm CMI $I_1(A:C|B)$ from the continuity of the trace norm CMI.
\begin{lem}[Continuity of trace norm CMI]\label{cont trace norm CMI}
If $\rho_{ABC} = \sigma_{ABC}+G_{ABC}$, and the deviation is traceless on system $C$, $Tr_C(G_{ABC})=0$, and bounded as $\lV G_{ABC}\rV_1 \le \epsilon$, then 
\begin{equation}
|I_1(A:C|B)_\rho -I_1 (A:C|B)_\sigma |\le 2 \epsilon.
\end{equation}
\end{lem}
\begin{proof}
All we need is the triangle inequality, which makes the analysis much simpler than Lemma~\ref{cont of CMI}. 
\begin{align}
&| I_1(A:C|B)_{\rho} - I_1(A:C|B)_{\sigma}|\\
&=| \lV\sigma_{ABC} + G_{ABC}- \sigma_A\otimes\sigma_{BC} - \sigma_A\otimes G_{BC}\rV_1 - \lV\sigma_{AB} - \sigma_A\otimes\sigma_{B}\rV_1 |\\
&-( \lV\sigma_{ABC} - \sigma_A\otimes\sigma_{BC}\rV_1 - \lV\sigma_{AB} - \sigma_A\otimes\sigma_{B}\rV_1 ) \\
&\le  \lV G_{ABC}- \sigma_A\otimes G_{BC}\rV_1 \\
& \le \lV G_{ABC}- \sigma_A\otimes G_{BC}\rV_1 \le 2  \lV G_{ABC}\rV_1.
\end{align}  The C-traceless assumption reduced the expression that $\rho_{AB} = \sigma_{AB}$.
\end{proof}
To show Proposition~\ref{prop:also trace norm CMI}, choosing $\sigma_{ABC}=\rho_{AB}\ot\tau_C$ for Theorem~\ref{thm: bistochastic}, $\sigma_{ABC}=\rho_{AB}\ot\nu_C$ for Theorem~\ref{thm: partially invariant} in the above we complete the proof.

\subsection{Remaining proof of Theorem~\ref{mainthm}: decay of $B$-measured conditional mutual information in MPDO}\label{sec:remaining proof measured}

In subsection~\ref{hilbert}, we provide the sufficient background for applying Hilbert's projective metric to show Condition 2 implies exponential decay of CMI; starting from subsection~\ref{sec:cond1} are the details for $1\implies 2$.

\subsubsection{The Hilbert's projective metric and proof of Lemma~\ref{lem: hilbert to trace}}
\label{hilbert}\label{proof cond 2}
The CP-self maps $\CM_{s_k}$ arising from measurement are not trace-preserving, hindering it difficult to approach from typical quantum information tools. It turns out the Hilbert's projective metric is suitable for this purpose, as it is designed to work for the set of all unnormalized states $S_+:=\{\rho \in \cB(\CH)| \rho \ge 0, \rho \ne 0\}$. While the general theory applies to convex cones, we will focus on the quantum case $S_+$, following partly the ideas in ~\cite{doi:10.1063/1.3615729}. \begin{defn}[Hilbert's projective metric]
\begin{align}
\forall a, b \in S, h(a,b): = \ln (\sup(a/b)\sup(b/a)) \\
\sup(a/b):= \inf\{\lambda\in \mathbb{R}| a\le \lambda b\} \label{supab}.
\end{align}

\end{defn}
Note that $\sup(a/b)\ne \sup(b/a)$ and the direction of inequality is important in \eqref{supab}. The metric is projective $h(a,b) = h(\alpha a,b)$, and it become a true metric when restricted to set of density operators, i.e., quotient out scalar multiples. This definition works for any \textit{proper cone}, and only implicitly depends on the actual geometry of the cone. 
In this metric, every positive maps $\CM: S_+ \rightarrow S_+ $ is contracting \footnote{In fact the unique metric for this to hold~\cite{Kohlberg1982}}. 
\begin{thm}[Birkhoff-Hopf contraction theorem {\cite[Thoerem~4]{doi:10.1063/1.3615729} }]
For all $\CM: S_+\rightarrow S_+$, the upper bound on contraction ratio is 
\begin{align}
\eta_\CM := \sup \limits_{a,b \in S_+} \frac{h(\CM(a),\CM(b))}{h(a,b)} = \tanh(\frac{\Delta (\CM)}{4}),
\end{align}
where $\Delta (\CM)$ is the projective diameter 
\begin{align} \label{projdiam}
\Delta (\CM) :=  \sup \limits_{a,b \in S_+} h(\CM(a),\CM(b)).
\end{align}
\end{thm}
Note that $\Delta (\CM) < \infty$ would imply strict contraction $\eta_\CM<1$. 

We eventually convert back to the norm via the following bound.
\begin{prop}[{\cite[Eq.~38]{doi:10.1063/1.3615729}}] \label{conversion}
For normalized density operators $\rho_1, \rho_2$, 
\begin{equation}
\frac{1}{2}\lVert \rho_1 - \rho_2\rVert_1 \le \tanh\left(\frac{h(\rho_1,\rho_2)}{4}\right) .
\end{equation}
\end{prop}

The above are the backgounds we need to prove the following lemma.
\begin{lem}
If all $\{\CM_{1}, \CM_{2}, \cdots\CM_{n}\}$ are CP maps that map any state to a full rank state, then for arbitrary sequence of $b = s_\ell,\cdots, s_1 \in \{1,\cdots, n\}^\ell$, it holds that
\begin{align}
\left\lV \frac{\CM_b[\rho_1]}{\tr(\CM_b[\rho_1])} - \frac{\CM_b[\rho_2]}{\tr(\CM_b[\rho_2])} \right\rV_1 = O(e^{-c \ell}), 
\end{align}

where the exponent $c>0$ is independent of $b$. 
\end{lem}
\begin{proof}
For each $\CM_{s}$, the distance between any two input states is finite because the image is full rank
$h(\CM_{s}(a),\CM_s(b)) < \infty.$
Hence the projective diameter as a supremum over compact set is also finite, and each $\CM_s$ is strictly contracting. Maximizing over $s = 1,\cdots, n$ provides a global contraction ratio bound $\eta<1$. 
\begin{align}
\Delta (\CM_s) &:=  \sup \limits_{a,b \in S_+} h(\CM_s(a),\CM_s(b)) < \infty \\
\forall s,\ \eta_{\CM_s} &= \tanh\left(\frac{\Delta (\CM_s)}{4}\right) \leq \eta< 1 .    
\end{align}

Then for all $a, b \in S_+$
\begin{align}
h(\CM_{s_\ell,\cdots, s_1}(a), \CM_{s_\ell,\cdots, s_1}(b))\le h(\CM_{s_1}(a),\CM_{s_1}(b))e^{-c(\ell-1)} \le \sup \limits_{s}( \Delta(\CM_{s}))\ e^{-c(\ell-1)} = O(e^{-c\ell}),
\end{align}
where again we used the fact the projective diameter is finite.  Finally we convert to the trace norm
\begin{align}
\left\lV \frac{\CM_b[\rho_1]}{\tr(\CM_b[\rho_1])} - \frac{\CM_b[\rho_2]}{\tr(\CM_b[\rho_2])} \right\rV_1 &\le 2\tanh\left(\frac{h(\CM_b[\rho_1],\CM_b[\rho_2])}{4}\right) \\
&\le O(e^{-c\ell}),
\end{align}
where $\tanh(x)\approx x$ for small $x$. 
\end{proof}

\subsubsection{Proof of Lemma~\ref{forget to MI} }
\begin{lem}
\label{recap:forget to MI}
Suppose the CP map $\CM: C' \rightarrow C$ is contractive in the sense that
\begin{equation}
    \lV \frac{\CM_b[\rho_1]}{\tr(\CM_b[\rho_1])} - \frac{\CM_b[\rho_2]}{\tr(\CM_b[\rho_2])} \rV_1 \le \epsilon.
\end{equation}
Then the state $\rho_{AC} := \CM[\sigma_{AC'}]$ is close to the product state 
\begin{align}
2T_b := \lV \rho_{AC,b} - \rho_{A,b}\otimes\rho_{C,b} \rV_1 \le 4d_{C'}\lV \sigma_{C'}^{-1}\rV_\infty \epsilon.
\end{align}
\end{lem}
\begin{proof}
 Bounded factors depending on the hidden system $C'$ may show up here and there, but it poses no threat when the error $\epsilon$ is exponentially small. 
\begin{align}
2T_b :&= \lV \rho_{AC,b} - \rho_{A,b}\otimes\rho_{C,b} \rV_1\\
&=\left\lV \frac{\CM_b[\sigma_{AC'}]}{\tr(\CM_b[\sigma_{C'}])} - \frac{Tr_C(\CM_b[\sigma_{AC'}])}{\tr(\CM_b[\sigma_{C'}])} \otimes \frac{\CM_b[\sigma_{C'}]}{\tr(\CM_b[\sigma_{C'}])}\right\rV_1 \\
&\le \left\lV \frac{\CM_b[\cdot ]}{\tr(\CM_b[\sigma_{C'}])} - \frac{Tr_C(\CM_b[\cdot])}{\tr(\CM_b[\sigma_{C'}])} \otimes \frac{\CM_b[\sigma_{C'}]}{\tr(\CM_b[\sigma_{C'}])} \right\rV_\diamond\\
&\le \frac{d_{A}}{\tr(\CM_b[\sigma_{C'}])}\sup_{\lV X\rV_1\le 1, X \subset {\rm supp}(\sigma_{C'})}\left\lV\CM_b[X] - \tr(\CM_b[X]) \frac{\CM_b[\sigma_{C'}]}{\tr(\CM_b[\sigma_{C'}])} \right\rV_1 \\
&\le \frac{4d_{A}}{\tr(\CM_b[\sigma_{C'}])}\sup_{\rho_1}\left\lV\CM_b[\rho_1] - \tr(\CM_b[\rho_1]) \frac{\CM_b[\sigma_{C'}]}{\tr(\CM_b[\sigma_{C'}])} \right\rV_1 \\
& \le 4d_{C'}\frac{\sup_{\rho_1\subset {\rm supp}(\sigma_{C'})}( \tr(\CM_b[\rho_1]) )}{\tr(\CM_b[\sigma_{C'}])} \epsilon \\
&\le 4d_{C'}\lV \sigma_{C'}^{-1}\rV_\infty \epsilon,
\end{align}
where in the second and third inequality we used that diamond norm is bounded by d times the superoperator-norm (Lemma~\ref{lem:cb to 1-1}) and the extra factor of 4 comes from turning $X$ into density operator $\rho_1$; in the fourth inequality we used the assumption; in the last inequality we used that $\lV \sigma_{C'}^{-1}\rV_\infty \sigma_{C'} \ge I \ge \rho_1$. Note that we are working on the support of $\sigma_{C'}$ so that the inverse is well-defined. 
\end{proof}

$2 \implies 3$ by setting $\sigma_{AC'}$ to be the maximally entangled state on $A\bar A$, which yields $\lV \rho_{AC,b} - \rho_{A,b}\otimes\rho_{C,b} \rV_1 \le 4d_{A}^2 \CO(e^{-c\ell})$. The conversion to mutual information straightforwardly follows from the AFW inequality,
\begin{align}
I(A:C)_b \le 2T_b\log(\min(d_A,d_C))+ (1+T_b)\log(1+T_b) - T_b\log(T_b) = O(e^{-c_1\ell}).\label{eq:AFW}
\end{align}
This immediately passes to CMI by taking expectation in Eq.~\eqref{eq:expectation over MI} and thus completes the proof.

\subsubsection{Proof of Lemma~\ref{fulldimkraus}}\label{sec:cond1}
\begin{lem}\label{recap:fulldimkraus}
 If ${\rm span}\{K_1, \cdots, K_p \} = Mat(D,D)$, then $\CN[\rho]: = \sum_i K_i \rho K_i^\dagger > 0, \forall \rho$.
\end{lem}
\begin{proof}
Suppose there exist $\ket{u},\ket{v}$, s.t. $\sum_j \bra{v}K_j \ket{u} \bra{u}K^\dagger_j \ket{v} = 0.$ Then we obtain
\begin{align}
|\bra{v}K_i\ket{u}|^2 = 0, \forall i\\
\implies \tr(K_i\ket{u}\bra{v}) = 0, \forall i.
\end{align}
This is a contradiction because $\{K_i\}$ are full dimensional. 
\end{proof}

\subsubsection{Proof of Lemma~\ref{increase span}}
The idea is rooted from a theorem of Burnside about algebra generated by matrices and the simultaneous invariant subspace.
\begin{thm}[Burnside~\cite{LOMONOSOV200445}]\label{burnside}
Consider $m_1, \cdots m_{p-1} \in Mat_\BC(D,D) $, acting on vector space $\BC^D$. 
The following are equivalent.
\begin{enumerate}
\item The algebra generated by $(m_1, \cdots, m_{p-1})$ is the full matrix algebra $Mat_\BC (D,D)$
\item For all non-trivial subspace $V\subset \BC^D$, 
\begin{equation}
  m_q V \subset V, \ \forall q \implies V = \BC^D,    
\end{equation}
i.e., the simultaneous invariant subspace of $m_1$ and $m_2$ is trivial or the whole space.
\end{enumerate}
\end{thm}
In our version, the vector space is $Mat_\BC(D,D)$, where Kraus operators live :

\begin{prop}\label{tweakburnside}
Consider $m_1, \cdots, m_{p-1} \in Mat_\BC(D,D) $, acting on vector space $Mat_\BC(D,D)$ by left multiplication. The following are equivalent.
\begin{enumerate}
\item The algebra generated by $(m_1,\cdots, m_{p-1})$ is the full matrix algebra $Mat_\BC (D,D)$
\item For all non-trivial subspace(as vector space) $W\subset Mat_\BC(d,d)$ containing an invertible element $B$, 
$$m_q W \subset  W, \forall q \implies W = Mat_\BC(D,D).$$
\end{enumerate}
\end{prop}
\begin{proof}
$(\implies)$ For all non-trivial subspace $W\in Mat_\BC(D,D)$, if $m_q W \subset  W, \forall q $, then same is true under left multiplication
\begin{align}
 (m_1,\cdots, m_{p-1} ) W \subset W.
\end{align}
Then in particular $W$ contains $ (m_1,\cdots, m_{p-1} ) B$, thus must be the full matrix algebra $W = Mat_\BC (D,D)$.

$(\impliedby)$ For a contradiction, suppose $ (m_1,\cdots, m_{p-1} )  \ne Mat_\BC (D,D)$. Then by theorem~\ref{burnside}, $ (m_1,\cdots, m_{p-1} ) $ is reducible with an proper invariant subspace $V\subsetneq \BC^D$. Consider the basis for which the first entries are basis vectors spanning $V$, i.e. $ (m_1,\cdots, m_{p-1} ) $ has some zeros at the left down corner:
$$
 (m_1,\cdots, m_{p-1} )  = \begin{bmatrix}
M_V & N \\
 0 & M_{V^c} 
\end{bmatrix}.
$$
Consider the subspace $W': =  (m_1,\cdots, m_{p-1} ) + \{\lambda I\}$ by adding the identity, which is invertible. 

Then the resulting subspace $W'$ is an invariant subspace
\begin{align}
m_qW'  &= m_q (m_1,\cdots, m_{p-1} )  + m_q \subset W'. \\
\end{align}
We arrive at contradiction with $W' =  (m_1,\cdots, m_{p-1} ) + \{\lambda I\}\ne Mat_\BC (D,D) $.
\end{proof}

We can now prove Lemma~\ref{increase span}:
\begin{lem}[Condition 1 implies the increment of span of Kraus operators] \label{recap:increase span}
Consider a CP self-map $\CM$ with $p>2$ Kraus operators $\{ E_1,\cdots, E_p\} $ such that $E_p$ is invertible, and $\left((E_p)^{-1} E_1,\cdots,(E_p)^{-1}  E_{p-1} \right)$ generate full matrix algebra $Mat_\BC(D,D)$ by addition and multiplication. Then for any set of Kraus operators $\{T_j\}$ not full rank and containing an invertible element, the dimension of span must increase after applying $\CM$
\begin{equation}
{\rm dim} ({\rm span}\{E_i T_j\}) > {\rm dim} ({\rm span}\{T_j\}). 
\end{equation}
\end{lem}
\begin{proof}
Left multiplying $E_i$ yields the span
$$
    \{ E_1T_j, \cdots ,E_pT_j\}.
$$ 
Notice that ${\rm dim}(\{T_j\}) = {\rm dim}(E_p\{T_j\})$ because $E_p$ is invertible, so if any of $E_1\{T_j\}, \cdots, E_{p-1}\{T_j\}$ is not a subspace of $E_p\{T_j\}$, then the dimension would increase ${\rm dim} ({\rm span}\{E_i T_j\}) > {\rm dim} ({\rm span}\{T_j\})$;
we only need to worry if 
\begin{equation}
E_1\{T_j\}, \cdots E_{p-1}\{T_j\} \subset E_p\{T_j\}.    
\end{equation}
Invert $E_p$ and by Proposition~\ref{tweakburnside}, $\{T_j\}$ must reach full rank already.
\begin{equation}
\{T_j\} = Mat_\BC (D,D).    
\end{equation}
\end{proof}
\subsubsection{Proof of Proposition~\ref{prop:algorithm}} \label{proof:algorithm}
It is simple to generate an algebra from set of matrices $m_1, \ldots, m_{p-1}$. Start with $S= {\rm span}\{m_1\}$, repeat the following steps:
\begin{enumerate}
    \item Add left multiplied matrices to the set, resulting $S' = \{S,  m_1S, \ldots, m_{p-1}S\}$.
    \item Find a linear basis for $S'$.\footnote{In the actual code, there always need to be a small error threshold to decide the linear independence.} If ${\rm dim}(S') = {\rm dim}(S)$, then terminate, and  ${\rm span}\{S'\}$ is the algebra generated by $m_1, \ldots, m_{p-1}$. If ${\rm dim}({\rm span}\{S'\}) = D^2$, then it is the full matrix algebra.  
\end{enumerate}
Checking linear independence uses polynomial runtime.

\section{Conclusion and discussions}
Motivated by the problem of showing the existence of the local parent Hamiltonians of MPDO, we have studied the CMI of MPDO. We have shown that MPDO constructed by bistochastic Y-Shaped channels with trivial correctable algebra have exponentially decaying CMI  and thus have approximately local parent Hamiltonians. We have shown a similar bound for a slightly more general class of channels under the restriction that certain decay constants are sufficiently small. We have also shown the exponential decay of CMI for Y-shaped channels with a forgetful component. We have introduced a trace norm variant of the CMI and have shown that for the above cases they obey no worse bounds than the CMI.

For more general Y-shaped channels, we have conjectured the completely contractive DPI (Conjecture~\ref{ccDPI}). We have shown that if this conjecture is true, every MPDO constructed by a Y-shaped channel with trivial correctable algebra has exponentially decaying CMI.  
For the measured MPDO, we have provided sufficient conditions implying the exponential decay of CMI. We have numerically confirmed (up to small bond dimension) that these conditions are generically true if the Y-shaped channel is generated by a Haar-random unitary.

Our results Theorem~\ref{thm: bistochastic} and Theorem~\ref{thm: partially invariant} only work for a restricted family of Y-shaped channels. The proof relies on the fact that the corresponding MPDO are approximately a product state. This is no longer true for general channels and the deviation from the product state do not obviously contract \footnote{ Bistochastic channels are rather special that it contracts all $p$-norm. See~\cite{doi:10.1063/1.2218675}  to see how this fails for non-bistochastic channels.}. A possible solution to avoid the structure of many-body state is resorting to our Conjecture~\ref{ccDPI}.
Note that having trivial correctable algebra is still only a sufficient condition, and a necessary and sufficient condition for exponentially decaying CMI is unclear yet.

Analysis of the CMI for the measured MPDO could be massively easier than the unmeasured case due to losing entanglement with $B$. However, our results on measured MPDO are still limited and only provide sufficient conditions. 
In contrast to the single channel quantum Wielandt's inequality~\cite{5550282} with sufficient and necessary conditions, Condition 1 in Theorem~\ref{mainthm} guarantees strict positivity for all sequences multiplicatively generated by a finite set of CP-maps $\{\CM_s\}$. Quantum Wielandt's inequality has a classical analog in matrix theory, however this multiple-channel  generalization has limited results even in the classical case (see e.g.,~\cite{PROTASOV2012749} for a related result). 

\section{Acknowledgement}
We thank Jean-Francois Quint for comments on multiplicative ergodic theory. We thank Mario Berta, Marco Tomamichel, Hao-Chung Cheng for discussions about the DPI for CMI. CFC is thankful for Physics TA Relief Fellowship and the Physics TA Fellowship at Caltech. KK acknowledges funding provided by the Institute for Quantum Information and
Matter, an NSF Physics Frontiers Center (NSF Grant {PHY}-{1733907}) and MEXT Quantum Leap Flagship Program (MEXT Q-LEAP) Grant Number JPMXS0120319794. FB acknowledges funding from NSF.

\newpage
\appendix

\section{Equivalence between correctable algebra and the Petz recovered map for bistochastic channels}\label{sec:Petz and correctable}
The structure of CPTP self-map, the invariant subspace, and the fixed-point algebra of the dual has been studied.
\begin{thm}[Combination of \cite{Hayden2004,PhysRevA.66.022318,lindbald1999}]\label{invar}
For every CPTP self-map $\CZ: S(\BC^d)\rightarrow S(\BC^d)$, the following holds:
\begin{enumerate}
\item The invariant subspace $S_{\CZ^\dagger}$ of $\CZ^\dagger$ forms a subalgebra $M\subset B(\CH_d)$, with factors $\alpha$.
\begin{equation}
S_{\CZ^\dagger} = \bigoplus_\alpha B(\CH_\alpha)\otimes I_{\bar{\alpha}}
\end{equation}
\item The invariant subspace of $S_{\CZ}$ of $\CZ$ has form 
\begin{equation}
S_{\CZ} = \bigoplus_\alpha p_\alpha \rho_\alpha \otimes \sigma_{\bar{\alpha}}
\end{equation}
with $\sigma_{\bar{\alpha}}$ determined by $\CZ$, and $\rho_\alpha, p_\alpha$ are free.
\item The channel $\CZ$ restricted on the block diagonal entries has form (its acting on off-diagonal part is not as simple)
\begin{equation}
\CZ =  \bigoplus_\alpha \CI_{\alpha} \otimes \CZ_{\bar{\alpha }}
\end{equation}
and $Z_{\bar{\alpha }}$ has unique fixed point $\sigma_\alpha$
\end{enumerate}
\end{thm}
Though, when a channel has different input-output dimension, we alternatively consider the correctable algebra $\CA(\CE)$. It turned out to correspond to the invariant subspace of $Z = \CP_{\tau, \CE} \CE$, the original channel composed with the Petz recovery map with the maximally mixed state as the reference.
\begin{prop}[Recap of Proposition~\ref{invar_cor}]
The correctable algebra of a channel $\CE$ equals to the fixed-point algebra of the Petz-recovery channel composed with channel $\CP_{\tau, \CE}\circ\CE$ (and the dual $\CE^\dagger \circ \CP_{\tau, \CE}^\dagger$, due to being self-adjoint).
\begin{equation}
S_{\CE^\dagger \circ \CP_{\tau, \CE}^\dagger} = S_{\CP_{\tau, \CE}\circ\CE} = \CA(\CE)    
\end{equation}
\end{prop}
\begin{proof}
By~\cite[Theorem~1]{Chen_2020}, the Petz map with the maximally mixed reference state  is a universal subalgebra recovery map, i.e. it recovers any subalgebra that can be recovered from any channel $D$. 
\begin{equation}
D \circ \CE (\rho)|_a =\rho_a, \forall \rho \implies \CP_{\tau, \CE} \circ \CE (\rho)|_a =\rho_a, \forall \rho
\end{equation}
The invariant subspace of $\CE^\dagger \circ \CP_{\tau, \CE}^\dagger$ contains the correctable algebra, and the converse is by definition true.
 \begin{equation}
 \CA(\CE)\subset S_{\CE^\dagger \circ \CP_{\tau, \CE}^\dagger} \subset \CA(\CE)
 \end{equation}
 We conclude the proof by that the Petz-recovered channel is self-adjoint w.r.t. to the H.S. norm
 \begin{align}
     \CE^\dagger \circ \CP_{\tau, \CE}^\dagger &=  \CE^\dagger \circ [\CE[\tau]^{-1/2}\tau\CE[\cdot ]\CE[\tau]^{-1/2}]= \tau\CE^\dagger  [\CE[\tau]^{-1/2}\CE[\cdot ]\CE[\tau]^{-1/2}]\\
     &=\CP_{\tau, \CE}\CE.
 \end{align}
\end{proof}

This provides an alternative proof that the correctable algebra behaves nicely when tensored with auxiliary system $B$.
\begin{prop} \label{correctabletensor}
Suppose the correctable algebra is 
\begin{equation}\label{firsteq_correctable}
\CA(\CE) := \gamma = \bigoplus_{\beta} B(\CH_\beta)\otimes I_{\bar{\beta}} \subset B(\CH_{C})
\end{equation}
then
\begin{equation}\CA(\CI_B \otimes \CE) = \bigoplus_\beta B(\CH_{B\beta})\otimes I_{\bar{\beta}} = B(\CH_B) \otimes \gamma
\end{equation}
\end{prop}
\begin{proof}
 By Proposition~\ref{invar_cor}, we get 
\begin{align}
\CA(\CI_B \otimes \CE) = S_{\CI_B \otimes \CE^\dagger \circ \CP_{\tau_{BC}, \CI_B \otimes \CE}^\dagger} = S_{\CI_B \otimes( \CE^\dagger \circ \CP_{\tau_{C}, \CE}^\dagger)} 
\end{align}
where we can use structure theorem~\ref{invar} for $\CZ:=\CP_{\tau_{C}, \CE}\circ \CE$
\begin{align}
\CI_B \otimes( \CE^\dagger \circ \CP_{\tau_{C}, \CE}^\dagger) &= \CI_B \otimes \bigoplus_\beta \CI_{\beta} \otimes \CZ_{, \bar{\beta }} \\
& =  \bigoplus_\beta \CI_{B\beta} \otimes \CZ_{, \bar{\beta }}
\end{align}
The invariant subspace of $\CI_B \otimes \CZ^\dagger$ coincide with Eq.\eqref{firsteq_correctable} 
\end{proof}
\section{Structure of periodic MPDO generated by Y-shaped bistochastic channels. }\label{sec:periodic bistochastic MPDO}
\begin{thm}
Consider the following MPDO constructed by closed loop of Y-shaped channels
\begin{align}
\rho_{B_1\cdots B_\ell}:= \sum_{i,j}\tr( \ket{j} \bra{i} \CN \circ \cdots \CN [\ket{i} \bra{j}] ).
\end{align}
Then the state is close to the maximally mixed state
\begin{align}
\rho_{B_1\cdots B_\ell} = \tau_{B_1\cdots B_\ell}+ O_{B_1\cdots B_\ell}
\end{align}
up to an exponentially small global operator
\begin{equation}
    \lV O_{B_1\cdots B_\ell} \rV_1 =\CO(\eta^{\ell}),
\end{equation}
which implies the decay of the CMI for any tripartition $A'B'C'$
\begin{equation}
    I(A':C'|B') = \CO(\ell\eta^{\ell}).
\end{equation}
Where the contraction ratio $\eta$ is given as
\begin{equation}
\eta: = \limsup_{\rho_C \rightarrow \tau_C } \frac{D(\CN[\rho_C]\|\tau_{C'})}{D(\rho_C\|\tau_C)}\,.
\end{equation}
\end{thm}
\begin{proof}
We first break the loop by rewriting
\begin{align}
    \rho_{B_1\cdots B_\ell} &= \sum_{i,j}\tr( \ket{j} \bra{i} \CN \circ \cdots \CN [\ket{i} \bra{j}] ) \\
    &= \sum_{O_{ij}}\tr(O^\dagger_{ij} \CN \circ \cdots \CN [O_{ij}] )\label{eq:sumij},
\end{align}
where $O_{ij}$ are any complete basis of operators orthonormal w.r.t. the Hilbert Schimdt inner product. Then recall the decomposition of $\CN \circ \cdots \CN$ as in the open boundary case
\begin{align}
    \CN \circ \cdots \CN[O]&=  (1-\BP_K) \CN[O]\otimes \tau_3 \cdots \tau_\ell \\ 
    &+ (1-\BP_K)\CN\BP_K\CN[O]\otimes \tau_4 \cdots \tau_\ell\\
    &+ \cdots \\
    &+ \BP_K\CN\cdots\BP_K\CN\BP_K\CN[O].
\end{align}
Plugging into Eq.~\eqref{eq:sumij} and choosing $O_{ij}$ to split into the maxmially mixed $I/\sqrt{d}$ and the basis for orthogonal complement $\{K_n\}$. 
\begin{align}
\sum_{O_{ij}}\tr(O^\dagger_{ij} \CN \circ \cdots \CN [O_{ij}] ) &= \sum_{n}\tr(K^\dagger_n \CN \circ \cdots \CN [K_n] ) \\
&+  d\tr(\tau \CN \circ \cdots \CN [\tau] )\\
&= \tau_{B_1\cdots B_\ell}+ \sum_{n}\tr(K^\dagger_n \CN \BP_K\circ \cdots\BP_K \CN [K_n] ),
\end{align}
where the sum over traceless $K_n$ contains a single term since other terms have $\tau_\ell$ which is orthogonal to $K_n$. When $\CN$ has trivial correctable algebra, by Lemma~\ref{DPI 2norm} the global operator is exponentially decaying in the normalized H-S norm and hence the trace norm using Cauchy-Schwartz inequality.  Finally by continutity of the CMI (Lemma ~\ref{cont of CMI}), each tripartition must have CMI exponentially small w.r.t to total length $\ell$.
\end{proof}

\end{document}